\documentclass[runningheads]{llncs}
\usepackage{hyperref,latexsym,breakurl,amsmath,amssymb,color,cspsymb}
\usepackage[shortlabels]{enumitem}
\setlist{nolistsep}  
\DeclareSymbolFont{frenchscript}{OMS}{ztmcm}{m}{n}
\DeclareMathSymbol{\D}{\mathord}{frenchscript}{68}    
\DeclareMathSymbol{\F}{\mathord}{frenchscript}{70}    
\DeclareMathSymbol{\Pow}{\mathord}{frenchscript}{80}  
\DeclareMathSymbol{\R}{\mathrel}{frenchscript}{82}    
\newcommand{\Sec}[1]{Section~\ref{sec:#1}}
\newcommand{\df}[1]{Def.~\ref{df:#1}}

\newcommand{\pr}[1]{Prop.~\ref{pr:#1}}
\newcommand{\lem}[1]{Lem.~\ref{lem:#1}}

\newcommand{\tab}[1]{Table~\ref{tab:#1}}
\makeatletter
\def\comesfrom{\@transition\leftarrowfill}
\def\goesto{\@transition\rightarrowfill}
\def\ngoesto{\@transition\nrightarrowfill}
\def\Goesto{\@transition\Rightarrowfill}
\def\nGoesto{\@transition\nRightarrowfill}
\def\xmapsto{\@transition\mapstofill}
\def\nxmapsto{\@transition\nmapstofill}
\def\@transition#1{\@@transition{#1}}
\newbox\@transbox
\newbox\@arrowbox
\newbox\@downbox
\def\@@transition#1#2%
   {\setbox\@transbox\hbox
      {\vrule height 1.5ex depth .8ex width 0ex\hskip0.25em$\scriptstyle#2$\hskip0.25em}
   \ifdim\wd\@transbox<1.5em
      \setbox\@transbox\hbox to 1.5em{\hfil\box\@transbox\hfil}\fi
   \setbox\@arrowbox\hbox to \wd\@transbox{#1}
   \ht\@arrowbox\z@\dp\@arrowbox\z@
   \setbox\@transbox\hbox{$\mathop{\box\@arrowbox}\limits^{\box\@transbox}$}
   \dp\@transbox\z@\ht\@transbox 10pt
   \mathrel{\box\@transbox}}
\def\nrightarrowfill{$\m@th\mathord-\mkern-6mu%
  \cleaders\hbox{$\mkern-2mu\mathord-\mkern-2mu$}\hfill
  \mkern-6mu\mathord\not\mkern-2mu\mathord\rightarrow$}
\def\Rightarrowfill{$\m@th\mathord=\mkern-6mu%
  \cleaders\hbox{$\mkern-2mu\mathord=\mkern-2mu$}\hfill
  \mkern-6mu\mathord\Rightarrow$}
\def\nRightarrowfill{$\m@th\mathord=\mkern-6mu%
  \cleaders\hbox{$\mkern-2mu\mathord=\mkern-2mu$}\hfill
  \mkern-6mu\mathord\not\mathord\Rightarrow$}
\def\mapstofill{$\m@th\mathord\mapstochar\mathord-\mkern-6mu%
  \cleaders\hbox{$\mkern-2mu\mathord-\mkern-2mu$}\hfill
  \mkern-6mu\mathord\rightarrow$}
\def\nmapstofill{$\m@th\mathord\mapstochar\mathord-\mkern-6mu%
  \cleaders\hbox{$\mkern-2mu\mathord-\mkern-2mu$}\hfill
  \mkern-6mu\mathord\not\mkern-2mu\mathord\rightarrow$}
\makeatother 
\newcommand{\goto}[1]{\stackrel{#1}{\longrightarrow}} 
\newcommand{\dto}[1]{\mathrel{\Goesto{\raisebox{.08em}{\scriptsize$#1$}}}}
\newcommand{\denote}[1]{\mbox{$[\hspace{-1.6pt}[$}\,#1\,\mbox{$]\hspace{-1.6pt}]$}}  
\newcommand{\pair}[1]{( #1 )}                         
\newcommand{\depth}{{\it depth}}                      
\newcommand{\Id}{{\it Id}}                            
\newcommand{\Th}{{\it Th}}                            
\newcommand{\throw}{\mathbin{\Theta\!_A}}             
\newcommand{\conceal}{\backslash}                     
\newcommand{\X}{p}                                    
\newcommand{\dcup}{\stackrel{\mbox{\huge .}}{\cup}}   
\newcommand{\plat}[1]{\raisebox{0pt}[0pt][0pt]{#1}}   
\newcommand{\weg}[1]{}                                
\begin{document}

\title{A Branching Time Model of CSP}

\author{
Rob van Glabbeek \inst{1,2}
}
\institute{
Data61, CSIRO, Sydney, Australia
\and
Comput. Sci. and Engineering, University of New South Wales, Sydney, Australia
}

\maketitle
\setcounter{footnote}{0}

\begin{abstract}
I present a branching time model of CSP that is finer than all other models of CSP proposed thus far.
It is obtained by taking a semantic equivalence from the linear time -- branching time spectrum,
namely divergence-preserving coupled similarity, and showing that it is a congruence for the
operators of CSP\@. This equivalence belongs to the bisimulation family of semantic equivalences,
in the sense that on transition systems without internal actions it coincides with strong bisimilarity.
Nevertheless, enough of the equational laws of CSP remain to obtain a complete axiomatisation for
closed, recursion-free terms.
\end{abstract}

\section{Introduction}

The process algebra CSP---\emph{Communicating Sequential Processes}---was presented in
{\sc Brookes, Hoare \& Roscoe} \cite{BHR84}.  It is sometimes called \emph{theoretical} CSP,
to distinguish it from the earlier language CSP of {\sc Hoare}~\cite{Ho78}.
It is equipped with a denotational semantics, mapping each CSP process to an element of the
failures-divergences model \cite{BHR84,BR85}. The same semantics can also be presented
operationally, by mapping CSP processes to states in a labelled transition system (LTS), and then mapping
LTSs to the failures-divergences model.
{\sc Olderog \& Hoare} \cite{OH86} shows that this yields the same result.
Hence, the failures-divergences model of CSP can alternatively be seen as a semantic equivalence on
LTSs, namely by calling two states in an LTS equivalent iff they
map to the same element of the failures-divergences model.

Several other models of CSP are presented in the literature, and each can be cast as a
semantic equivalence on LTSs, which  is a congruence for the operators
of CSP\@. One such model is called \emph{finer} than another if its associated equivalence relation is
finer, i.e., included in the other one, or more discriminating.
The resulting hierarchy of models of CSP has two pillars: the divergence-strict models, most of
which refine the standard failures-divergences model, and the stable models, such as the model based
on stable failures equivalence from {\sc Bergstra, Klop \& Olderog} \cite{BKO87}, or the
stable revivals model of {\sc Roscoe} \cite{Ro09}.

Here I present a new model, which can be seen as the first branching time model of CSP, and the first
that refines all earlier models, i.e.\ both pillars mentioned above. It is based on the notion of
coupled similarity from {\sc Parrow \& Sj\"odin} \cite{PS92}. What makes it an interesting
model of CSP---as opposed to, say, strong or divergence-preserving weak bisimilarity---is
that it allows a complete equational axiomatisation for closed recursion-free CSP processes that fits
within the existing syntax of that language.

\section{CSP}\label{sec:CSP}

CSP \cite{BHR84,BR85,Ho85} is parametrised with a set $\Sigma$ of {\em communications\/}.
In this paper I use the subset of CSP given by the following grammar.
\[
P,Q ::=
\begin{array}[t]{@{}l@{}}
\STOP \mid \div \mid a\rightarrow P \mid P \intchoice Q \mid P \extchoice Q \mid P \timeout Q \mid
\\
P \|_A Q \mid P \backslash A \mid f(P) \mid P\interrupt Q \mid P \throw Q
\mid \X \mid \mu \X.P
\end{array}
\]
Here $P$ and $Q$ are CSP expressions,
$a\in\Sigma$, $A\subseteq\Sigma$ and $f:\Sigma\rightarrow\Sigma$.
Furthermore, $\X$ ranges over a set of \emph{process identifiers}.
A CSP \emph{process} is a CSP expression in which each occurrence of a process identifier $\X$
lays within a recursion construct $\mu\X.P$.
The operators in the above grammar are \emph{inaction}, \emph{divergence},
\emph{action prefixing}, \emph{internal}, \emph{external} and \emph{sliding choice},
\emph{parallel composition}, \emph{concealment}, \emph{renaming}, \emph{interrupt} and \emph{throw}.
Compared to \cite{Ro97,Ro10}, this leaves out
\begin{itemize}
  \item successful termination ($\SKIP$) and sequential composition (;),
  \item infinitary guarded choice,
  \item prefixing operators with name binding, conditional choice,
  \item relational renaming, and
  \item the version of internal choice that takes a possibly infinite
    set of arguments.
\end{itemize}
The operators $\STOP$, $a\rightarrow$, $\intchoice$, $\extchoice$, $\conceal A$, $f(\_)$ and recursion
stem from \cite{BHR84}, and $\div$ and $\|_A$ from \cite{OH86}, whereas $\timeout$,
$\interrupt$ and $\throw$ were added to CSP by {\sc Roscoe} \cite{Ro97,Ro10}.

\begin{table}[htb]
\begin{center}
\framebox{$\begin{array}{ccccc}
\div\goto{\tau}\div &
(a\rightarrow P) \goto{a} P &
P \intchoice Q \goto{\tau} P &
P \intchoice Q \goto{\tau} Q \\[2ex]
\displaystyle\frac{P\goto{a} P'}{P\extchoice Q \goto{a} P'} &
\displaystyle\frac{P\goto{\tau} P'}{P\extchoice Q \goto{\tau} P'\extchoice Q} &
\displaystyle\frac{Q\goto{a} Q'}{P\extchoice Q \goto{a} Q'} &
\displaystyle\frac{Q\goto{\tau} Q'}{P\extchoice Q \goto{\tau} P\extchoice Q'} \\[4ex]
\displaystyle\frac{P\goto{a} P'}{P\timeout Q \goto{a} P'} &
\displaystyle\frac{P\goto{\tau} P'}{P\timeout Q \goto{\tau} P'\timeout Q} &
P \timeout Q \goto{\tau} Q &
\displaystyle\frac{P \goto{\alpha} P'}{f(P) \goto{f(\alpha)} f(P')} \\[4ex]
\displaystyle\frac{P\goto{\alpha} P'~~{\scriptstyle(\alpha\notin A)}}{P\|_AQ \goto{\alpha} P'\|_AQ} &
\multicolumn{2}{c}{
\displaystyle\frac{P\goto{a} P'~~Q\goto{a} Q'~~{\scriptstyle(a\in A)}}{P\|_AQ \goto{a} P'\|_AQ'}} &
\displaystyle\frac{Q\goto{\alpha} Q'~~{\scriptstyle(\alpha\notin A)}}{P\|_AQ \goto{\alpha} P\|_AQ'} \\[4ex]
\displaystyle\frac{P \goto{\alpha} P'~~{\scriptstyle(\alpha\notin A)}}{P\conceal A \goto{\alpha} P'\conceal A} &
\displaystyle\frac{P \goto{a} P'~~{\scriptstyle(a\in A)}}{P\conceal A \goto{\tau} P'\conceal A} &
\displaystyle\frac{P\goto{\alpha} P'~~{\scriptstyle(\alpha\notin A)}}{P\throw Q \goto{a} P'\throw Q} &
\displaystyle\frac{P\goto{a} P'~~{\scriptstyle(a\in A)}}{P\throw Q \goto{a} Q}\\[4ex]
\displaystyle\frac{P\goto{\alpha} P'}{P\triangle Q \goto{\alpha} P'\triangle Q} &
\displaystyle\frac{Q\goto{\tau} Q'}{P\triangle Q \goto{\tau} P'\triangle Q'} &
\displaystyle\frac{Q\goto{a} Q'}{P\triangle Q \goto{a} Q'} &
\multicolumn{2}{c}{\mu \X.P \goto{\tau} P[\mu \X.P/\X]}
\end{array}$}
\end{center}
\caption{Structural operational semantics of CSP}
\label{tab:CSP}
\end{table}

\noindent
The operational semantics of of CSP is given by the binary transition relations $\goto{\alpha}$
between CSP processes. The transitions $P \goto{\alpha} Q$ are derived by the rules in \tab{CSP}.
Here $a$, $b$ range over $\Sigma$ and $\alpha$, $\beta$ over \plat{$\Sigma\dcup\{\tau\}$}, and
relabelling operators $f$ are extended to \plat{$\Sigma\dcup\{\tau\}$} by $f(\tau)=\tau$.
The transition labels $\alpha$ are called \emph{actions}, and $\tau$ is the \emph{internal action}.

\section{The Failures-Divergences Model of CSP}

The process algebra CSP stems from {\sc Brookes, Hoare \& Roscoe} \cite{BHR84}.  It is also called
\emph{theoretical} CSP, to distinguish it from the language CSP of {\sc Hoare}~\cite{Ho78}.  Its semantics
\cite{BR85} associates to each CSP process a pair $\langle F, D\rangle$ of \emph{failures}
$F \subseteq \Sigma^* \times \Pow(\Sigma)$ and \emph{divergences} $D \subseteq \Sigma^*$,
subject to the conditions:
\begin{align}
  & \pair{\varepsilon,\emptyset} \in F \tag{N1}\\
  & \pair{st,\emptyset} \in F \Rightarrow \pair{s,\emptyset}\in F \tag{N2}\\
  & \pair{s,X} \in F \wedge Y\subseteq X \Rightarrow \pair{s,Y}\in F \tag{N3}\\
  & \pair{s,X} \in F \wedge \forall c\in Y .\, \pair{sc,\emptyset}\notin F \Rightarrow \pair{s,X\cup Y}\in F \tag{N4}\\
  & \forall Y \in\Pow_{\it fin}(X).\, \pair{s,Y} \in F \Rightarrow \pair{s,X}\in F \tag{N5}\\[1ex]
  & s\in D \Rightarrow st \in D \tag{D1}\\
  & s\in D \Rightarrow \pair{st,X}. \tag{D2}
\end{align}
Here $\varepsilon\in\Sigma^*$ is the empty sequence of communications and
$st$ denotes the concatenation of sequences $s$ and $t\in \Sigma^*$.
If $\langle F, D\rangle$ is the semantics of a process $P$, $\pair{s,\emptyset}\in F$, with $s\not\in D$,
tells that $P$ can perform the sequence of communications $s$, possibly interspersed with internal actions.
Such a sequence is called a \emph{trace} of $P$, and Conditions N1 and N2 say that the set of traces
of any processes is non-empty and prefix-closed. A failure $\pair{s,X}\in F$, with $s\notin D$, says
that after performing the trace $s$, $P$ may reach a state in which it can perform none of the
actions in $X$, nor the internal action. A communication $x\in \Sigma$ is thought to occur in
cooperation between a process and its environment. Thus $\pair{s,X}\in F$ indicates that deadlock
can occur if after performing $s$ the process runs in an environment that allows the execution of
actions in $X$ only. From this perspective, Conditions N3 and N4 are obvious.

A divergence $s \in D$ is a trace after which an infinite sequence of internal actions is possible.
In the failures-divergences model of CSP divergence is regarded \emph{catastrophic}: all further
information about the process' behaviour past a divergence trace is erased. This is accomplished by
\emph{flooding}: all conceivable failures $\pair{st,X}$ and divergences $st$ that have $s$ as a
prefix are added to the model (regardless whether $P$ actually has a trace $st$).

A CSP process $P$ from the syntax of \Sec{CSP} has the property that for any trace $s$ of $P$, with
$s\notin D$, the set ${\it next}(s)$ of actions $c$ such that $sc$ is also a trace of $P$ is finite.
By (N3--4), $\pair{s,X}\in F$ iff $\pair{s,X\cap {\it next}(s)}\in F$.
It follows that if $\pair{s,X}\notin F$, then there is a finite subset $Y$ of $X$, namely $X \cap {\it next}(s)$,
such that $\pair{s,Y}\notin F$. This explains Condition (N5).

In {\sc Brookes \& Roscoe} \cite{BR85} the semantics of CSP is defined denotationally: for each $n$-ary CSP operator $Op$,
a function is defined that extracts the failures and divergences of $Op(P_1,\dots,P_n)$ out of the
failures and divergences of the argument processes $P_1,\dots,P_n$. The meaning of a recursively
defined CSP process $\mu \X.P$ is obtained by means of fixed-point theory.
Alternatively, the failures and divergences of a CSP process can be extracted from its operational
semantics:

\begin{definition}\rm\label{df:FPoperational}
  Write $P \dto{} Q$ if there are processes $P_0,\dots,P_n$, with $n\geq 0$, such that $P=P_0$,
  $P_i \goto{\tau} P_{i+1}$ for all $0\leq i < n$, and $P_n=Q$.

  Write $P \dto{\alpha} Q$ if there are processes $P',Q'$ with
  $P \dto{} P' \goto{\alpha} Q' \dto{} Q$.

  Write $P \dto{\hat\alpha} Q$ if either $\alpha\in\Sigma$
  and $P \dto{\alpha} Q$, or $\alpha=\tau$ and $P \dto{} Q$.

  Write $P \dto{s} Q$, for $s=a_1a_2 \dots a_n \in \Sigma^*$ with $n\geq 0$, if there are processes
  $P_0,\dots,P_n$ such that $P=P_0$, $P_i \dto{a_i} P_{i+1}$ for all $0\leq i < n$, and $P_n=Q$.

  Let $I(P)=\{\alpha \in \Sigma \cup \{\tau\} \mid \exists Q.\, P \goto{\alpha} Q\}$.

  Write $P{\Uparrow}$ if there are processes $P_i$ for
  all $i\mathbin\geq 0$ with $P \dto{s} P_0 \goto{\tau} P_1 \goto{\tau} \dots$.

  $s\in \Sigma^*$ is a \emph{divergence trace} of a process $P$ if there is a $Q$ with $P \dto{s} Q{\Uparrow}$.

  The \emph{divergence set} of $P$ is $\D(P):=\{st\mid s \mbox{~is a divergence trace of~} P\}$.

  A \emph{stable failure} of a process $P$ is a pair $\pair{s,X} \in \Sigma^*\times\Pow(\Sigma)$ such
  that \plat{$P \dto{s} Q$} for some $Q$ with $I(Q)\cap (X\cup\{\tau\}) = \emptyset$.
  The \emph{failure set} of a process $P$ is
  $\F(p) = \{\pair{s,X} \mid s \in \D(P) \mbox{~or~}\pair{s,X} \mbox{~is a stable failure of~} P\}$.

  The semantics $\denote{P}_{\F\D}$ of a CSP process $P$ is the pair $\langle \F(P), \D(P) \rangle$.

  Processes $P$ and $Q$ are \emph{failures-divergences equivalent}, notation $P \equiv_{FD} Q$,
  iff $\denote{P}_{\F\D}=\denote{Q}_{\F\D}$.
  Process $P$ is a \emph{failures-divergences refinement} of $Q$, notation $P \sqsupseteq_{FD} Q$,
  iff $\F(P) \subseteq \F(Q) \wedge \D(P) \subseteq \D(Q)$.
\end{definition}
The operational semantics of \Sec{CSP} (then without the operators $\timeout$, $\interrupt$ and
$\throw$) appears, for instance, in \cite{OH86}, and was created after the denotational semantics.
In {\sc Olderog \& Hoare} \cite{OH86} it is shown that the semantics $\denote{P}$ of a CSP process defined operationally
through \df{FPoperational} equals the denotational semantics given in~\cite{BR85}.
The argument extends smoothly to the new operators $\timeout$, $\interrupt$ and $\throw$ \cite{Ro10}.
This can be seen as a justification of the operational semantics of \Sec{CSP}.

In {\sc Brookes, Hoare \& Roscoe} \cite{BHR84} a denotational semantics of CSP was given involving failures only.  Divergences
were included only implicitly, namely by thinking of a trace $s$ as a divergence of a process $P$
iff $P$ has all failures $\pair{st,X}$.\linebreak[2] So the semantics of $\div$ or $\mu X.X$ is simply the
set of all failure pairs.  As observed in {\sc De Nicola} \cite{DN85}, this approach invalidates a
number of intuitively valid laws, such as $P \extchoice \div = \div$. The improved semantics of
\cite{BR85} solves this problem.

In {\sc Hoare} \cite{Ho85} a slightly different semantics of CSP is given, in which a process is determined by
its failures, divergences, as well as its \emph{alphabet}. The latter is a superset of the set of
communications the process can ever perform. Rather than a parallel composition $\|_A$ for each set
of synchronising actions $A \subseteq \Sigma$, this approach has an operator $\|$ where the set of 
synchronising actions is taken to be the intersection of the alphabets of its
arguments. Additionally, there is an operator $\interleave$, corresponding to $\|_\emptyset$.
This approach is equally expressive as the one of \cite{BR85}, in the sense that there are semantics 
preserving translations in both directions. The work reported in this paper could just as well have
been carried out in this \emph{typed} version of CSP\@.

\section{A Complete Axiomatisation}

\newcommand{\Bot}{\mathbf{\bot}}
\newcommand{\Ii}{\mathbf{I1}}
\newcommand{\Ic}{\mathbf{I2}}
\newcommand{\Ia}{\mathbf{I3}}
\newcommand{\Io}{\mathbf{I4}}
\newcommand{\Ei}{\mathbf{E1}}
\newcommand{\Ec}{\mathbf{E2}}
\newcommand{\Ea}{\mathbf{E3}}
\newcommand{\Es}{\mathbf{E4}}
\newcommand{\Ed}{\mathbf{E5}}
\newcommand{\Si}{\mathbf{S1}}
\newcommand{\Sa}{\mathbf{S2}}
\newcommand{\SE}{\mathbf{S3}}
\newcommand{\SI}{\mathbf{S4}}
\newcommand{\Ss}{\mathbf{S5}}
\newcommand{\IS}{\mathbf{S6}}
\newcommand{\ES}{\mathbf{S7}}
\newcommand{\EI}{\mathbf{D1}}
\newcommand{\IE}{\mathbf{D2}}
\newcommand{\EA}{\mathbf{D3}}
\newcommand{\AI}{\mathbf{D4}}
\newcommand{\Sl}{\mathbf{SC}}
\newcommand{\Pru}{\mathbf{Prune}}
\newcommand{\Pa}{\mathbf{P0}}
\newcommand{\Pc}{\mathbf{P1}}
\newcommand{\PI}{\mathbf{P2}}
\newcommand{\PE}{\mathbf{P4}}
\newcommand{\PD}{\mathbf{P3}}
\newcommand{\HI}{\mathbf{H1}}
\newcommand{\HA}{\mathbf{H2}}
\newcommand{\HE}{\mathbf{H3}}
\newcommand{\HD}{\mathbf{H4}}
\newcommand{\CE}{\mathbf{H5}}
\newcommand{\CED}{\mathbf{H6}}
\newcommand{\CEI}{\mathbf{H7}}
\newcommand{\CEDI}{\mathbf{H8}}
\newcommand{\RS}{\mathbf{R0}}
\newcommand{\RI}{\mathbf{R1}}
\newcommand{\RE}{\mathbf{R2}}
\newcommand{\RA}{\mathbf{R3}}
\newcommand{\Rs}{\mathbf{R4}}
\newcommand{\RD}{\mathbf{R5}}
\newcommand{\TS}{\mathbf{T0}}
\newcommand{\TI}{\mathbf{T1}}
\newcommand{\TE}{\mathbf{T2}}
\newcommand{\TA}{\mathbf{T3}}
\newcommand{\AT}{\mathbf{T4}}
\newcommand{\Ts}{\mathbf{T5}}
\newcommand{\TD}{\mathbf{T6}}
\newcommand{\US}{\mathbf{U0}}
\newcommand{\UI}{\mathbf{U1}}
\newcommand{\UE}{\mathbf{U2}}
\newcommand{\UA}{\mathbf{U3}}
\newcommand{\Us}{\mathbf{U4}}
\newcommand{\UD}{\mathbf{U5}}

In \cite{BHR84,BR85,DN85,Ho85,Ro97,Ro10} many algebraic laws  $P=Q$, resp.\ $P \sqsubseteq Q$, are stated
that are \emph{valid} w.r.t.\ the failures-divergences semantics of CSP, meaning that $P \equiv_{FD} Q$,
resp.\ $P \sqsubseteq_{FD} Q$.
If $\Th$ is a collection of equational laws $P=Q$ then $\Th\vdash R=S$ denotes that the equation
$R=S$ is derivable from the equations in $\Th$ using reflexivity, symmetry, transitivity and the
rule of congruence, saying that if $Op$ is an $n$-ary CSP operator and $P_i = Q_i$ for $i=1,\dots,n$
then $Op(P_1,\dots,P_n)=Op(Q_1,\dots,Q_n)$. Likewise, if $\Th$ is a collection of inequational laws
$P\sqsubseteq Q$ then $\Th\vdash R\sqsubseteq S$ denotes that the inequation $R\sqsubseteq S$ is
derivable from the inequations in $\Th$ using reflexivity, transitivity and the rule saying that if
$Op$ is an $n$-ary CSP operator and $P_i \sqsubseteq Q_i$ for $i=1,\dots,n$ then
$Op(P_1,\dots,P_n)\sqsubseteq Op(Q_1,\dots,Q_n)$. 

\begin{definition}\rm
An equivalence $\sim$ on process expressions is called a \emph{congruence} for an  $n$-ary operator
$Op$ if $P_i \sim Q_i$ for $i=1,\dots,n$ implies $Op(P_1,\dots,P_n)\sim Op(Q_1,\dots,Q_n)$.
A preorder $\preceq$ is a \emph{precongruence} for $Op$, or $Op$ is \emph{monotone} for $\preceq$,
if $P_i \preceq Q_i$ for $i=1,\dots,n$ implies $Op(P_1,\dots,P_n)\preceq Op(Q_1,\dots,Q_n)$.
\end{definition}
If $\sim$ is a congruence for all operators of CSP (resp.\ $\preceq$ is a precongruence for all operators of CSP)
and $\Th$ is a set of (in)equational laws that are valid for $\sim$ (resp.\ $\preceq$) then any (in)equation
$R=S$ with $\Th\vdash R=S$ (resp.\ $R\sqsubseteq S$ with $\Th\vdash R\sqsubseteq S$) is valid for $\sim$ (resp.\ $\preceq$).

$\equiv_{FD}$ is a congruence for all operators of CSP\@. This follows
immediately from the existence of the denotational failures-divergences semantics.
Likewise, $\sqsubseteq_{FD}$ is a precongruence for all operators of CSP \cite{BHR84,BR85,DN85,Ho85,OH86,Ro97,Ro10}.

\begin{definition}\rm
A set $\Th$ of (in)equational laws---an \emph{axiomatisation}---is \emph{sound and complete} for an equivalence
$\sim$ (or a preorder $\preceq$) if $\Th \vdash R=S$ iff $R \sim S$ (resp.\ $\Th \vdash R\sqsubseteq S$
iff $R \preceq S$). Here ``$\Rightarrow$'' is \emph{soundness} and ``$\Leftarrow$'' completeness.
\end{definition}
In {\sc De Nicola} \cite{DN85} a sound and complete axiomatisation of $\sqsubseteq_{FD}$ for recursion-free CSP,
and no process identifiers or variables, is presented.
It is quoted in \tab{axsFD}.
As this axiomatisation consist of a mix of equations and inequations, formally it is an inequational
axiomatisation, where an equation $P=Q$ is understood as the conjunction of $P \sqsubseteq Q$ and
$Q \sqsubseteq P$. This mixed use is justified because $\equiv_{FD}$ is the \emph{kernel} of $\sqsubseteq_{FD}$:
one has $P \equiv_{FD} Q$ iff $P \sqsubseteq_{FD} Q \wedge Q \sqsubseteq_{FD} P$.

\begin{table}[p]
\[
\begin{array}{@{}l@{\quad}rcl@{}}
\textcolor{red}{\Bot}  &  \div              & \sqsubseteq & P                       \\[1ex]
\Ii  &  P \intchoice P              & = & P                       \\
\Ic  &  P \intchoice Q              & = & Q \intchoice P              \\
\Ia  &  P \intchoice (Q \intchoice R)   & = & (P \intchoice Q) \intchoice R   \\
\textcolor{red}{\Io}  &  P \intchoice Q              & \sqsubseteq & P                       \\[1ex]
\textcolor{red}{\Ei}  &  P \extchoice P                & = & P                       \\
\Ec  &  P \extchoice Q                & = & Q \extchoice P                \\
\Ea  &  P \extchoice (Q \extchoice R)       & = & (P \extchoice Q) \extchoice R       \\
\Es  &  P \extchoice \STOP            & = & P                       \\
\textcolor{red}{\Ed}  &  P \extchoice \div            & = & \div                       \\[1ex]
\EI  &  P \extchoice (Q \intchoice R)     & = & (P \extchoice Q) \intchoice (P \extchoice R) \\
\textcolor{red}{\IE}  &  P \intchoice (Q \extchoice R)     & = & (P \intchoice Q) \extchoice (P \intchoice R) \\
\textcolor{red}{\EA}  &  (a\rightarrow P) \extchoice (a\rightarrow Q)& = &a\rightarrow (P\intchoice Q)\\
\textcolor{red}{\AI}  & (a\rightarrow P) \intchoice (a\rightarrow Q) & = &a\rightarrow (P\intchoice Q)\\
\textcolor{red}{\Sl}  & P \timeout Q & = & (P\extchoice Q)\intchoice Q\\[1ex]
\Pa  &  P \|_A (Q \|_A R)       & = & (P \|_A Q) \|_A R        \\
\Pc  &  P \|_A Q                & = & Q \|_A P                 \\
\textcolor{red}{\PI}  &  (P \intchoice Q) \|_A R     & = & (P \|_A R) \intchoice (Q \|_A R) \\
\textcolor{red}{\PD}  &  P \|_A \div             & = & \div                       \\
\PE  & \multicolumn{3}{@{}l}{\mbox{~If~} P = \Extchoice_{i \in I}(a_i\rightarrow P_i)
  \mbox{~and~} Q = \Extchoice_{j \in J}(b_j\rightarrow Q_j) \mbox{~then}:}\\
     &  P \| Q                  & = &
                  \Extchoice_{a_i\notin A} (a_i \rightarrow (P_i\|_A Q)) \extchoice \\
             &&& \Extchoice_{a_j=b_j\in A} (a_i \rightarrow (P_i \|_A Q_j)) \extchoice\\
             &&& \Extchoice_{b_j\notin A} (b_j \rightarrow (P\|_A Q_j)) \\[2ex]
\HI  &  (P \intchoice Q)\conceal A   & = & (P\conceal A) \intchoice (Q\conceal A) \\
\textcolor{red}{\HA}  &  (P \extchoice a \rightarrow Q)\conceal A   & = & ((P\extchoice Q)\conceal A) \intchoice (Q\conceal A) \\
\HE  &  \big(\Extchoice_{i\in I} (b_i \rightarrow P_i )\big)\conceal A & = &
                 \big(\Extchoice_{i\in I} (b_i \rightarrow (P_i\conceal A) )\big)\qquad \mbox{if~}\forall i\in I.\, b_i\notin A\\
\HD  &  \div\conceal A         & = & \div   \\[1ex]
\RI  &  f(P \intchoice Q)         & = & f(P) \intchoice f(Q)   \\
\RE  &  f(P \extchoice Q)         & = & f(P) \extchoice f(Q)   \\
\RA  &  f(a \rightarrow P)        & = & f(a) \rightarrow f(P)   \\
\Rs  &  f(\STOP)                   & = & \STOP   \\
\RD  &  f(\div)                   & = & \div   \\[1ex]
\TI  &  (P \intchoice Q)\throw R         & = & (P\throw R) \intchoice (Q\throw R)   \\
\TE  &  (P \extchoice Q)\throw R         & = & (P\throw R) \extchoice (Q\throw R)   \\
\TA  &  (a \rightarrow P) \throw Q       & = & a \rightarrow (P\throw Q) \hfill\mbox{if $a \notin A$}  \\
\AT  &  (a \rightarrow P) \throw Q       & = & a \rightarrow Q \hfill \mbox{if $a \in A$}  \\
\Ts  &  \STOP \throw Q       & = & \STOP  \\
\TD  &  \div \throw Q       & = & \div  \\[1ex]
\UI  &  (P \intchoice Q)\interrupt R         & = & (P\interrupt R) \intchoice (Q\interrupt R)   \\
\textcolor{red}{\UE}  &  (P \extchoice Q)\interrupt R         & = & (P\interrupt R) \extchoice (Q\interrupt R)   \\
\textcolor{red}{\UA}  &  (a \rightarrow P) \interrupt Q       & = & (a \rightarrow (P\interrupt Q))\extchoice Q \\
\Us  &  \STOP \interrupt P                   & = & P \\
\textcolor{red}{\UD}  &  \div \interrupt P                    & = & \div
\vspace{5pt}
\end{array}
\]
\caption{A complete axiomatisation of $\sqsubseteq_{FD}$ for recursion-free CSP}\label{tab:axsFD}
\end{table}

In \cite{DN85}, following \cite{BHR84,BR85}, two parallel composition operators $\|$ and $\interleave$ were considered,
instead of the parametrised operator $\|_A$. Here $\|=\|_\Sigma$ and $\interleave = \|_\emptyset$.
In \tab{axsFD} the axioms for these two operators are unified into an axiomatisation of $\|_A$.
Additionally, I added axioms for sliding choice, renaming, interrupt and throw---these operators were not
considered in \cite{DN85}. The associativity of parallel composition (Axiom $\Pa$) is not included
in \cite{DN85} and is not needed for completeness. I added it anyway, because of its
importance in equational reasoning.

The soundness of the axiomatisation of \tab{axsFD} follows from $\sqsubseteq_{FD}$ being a
precongruence, and the validity of the axioms---a fairly easy inspection using the denotational
characterisation of $\denote{\_}$. To obtain completeness,
write \plat{$\Extchoice_{i\in I}P_i$}, with $I\mathbin=\{i_1,\ldots,i_n\}$ any finite index set, for
$P_{i_1} \mathbin{\extchoice} P_{i_2} \mathbin{\extchoice} \dots \mathbin{\extchoice} P_{i_n}$,
where \plat{$\Extchoice_{i\in \emptyset}P_i$} represents $\STOP$. This notation is justified by Axioms \textbf{E2--4}.
Furthermore,\vspace{-2pt}  $\Intchoice_{j\in J}P_j$,\vspace{1pt} with $J\mathbin=\{j_1,..,j_m\}$ any finite, nonempty index set, denotes
$P_{j_1} \mathbin{\intchoice} P_{j_2} \mathbin{\intchoice} \dots \mathbin{\intchoice} P_{j_m}$.
This notation is justified by Axioms $\Ic$ and $\Ia$.
Now a \emph{normal form} is a defined as a CSP expression of the form $\div$ or \plat{$\Intchoice_{j\in J}R_j$},
with \plat{$R_j=\big(\Extchoice_{k\in K_j}(a_{kj}\rightarrow R_{kj})\big)$}\vspace{2pt} for $j\in J$,
where the subexpressions $R_{kj}$ are again in normal form.
Here $J$ and the $K_j$ are finite index sets, $J$ nonempty.

Axioms $\Bot$ and $\Io$ derive $P \intchoice \div =\div$. Together with
Axioms $\EI$, $\Sl$, \mbox{\textbf{P1--4}}, \textbf{H1--4}, \textbf{R1--5}, \textbf{T1--6} and \textbf{U1--5}
this allows any recursion-free CSP expression to be rewritten into normal form.
In \cite{DN85} it is shown that for any two normal forms $P$ and $Q$ with $P \sqsubseteq_{FD} Q$,
Axioms $\Bot$, \textbf{I1--4}, \mbox{\textbf{E1--5}} and \mbox{\textbf{D1--4}} derive $\vdash P=Q$.
Together, this yields the completeness of the axiomatisation of \tab{axsFD}.

\section{Other Models of CSP}

Several alternative models of CSP have been proposed in the literature, including the
readiness-divergences model of {\sc Olderog \& Hoare} \cite{OH86} and the stable revivals model of
{\sc Roscoe} \cite{Ro09}.
A hierarchy of such models is surveyed in {\sc Roscoe}~\cite{Ro10}.
Each of these models corresponds with a preorder (and associated semantic equivalence) on labelled
transition systems. In \cite{vG93} I presented a survey of semantic equivalences and preorders on
labelled transition systems, ordered by inclusion in a lattice. Each model occurring in \cite{Ro10}
correspond exactly with with one of the equivalences of \cite{vG93}, or---like the stable revivals
model---arises as the meet or join of two such equivalences.

In the other direction, not every semantic equivalence or preorder from \cite{vG93} yields a sensible model of CSP\@.
First of all, one would want to ensure that it is a (pre)congruence for the operators of CSP\@.
Additionally, one might impose sanity requirements on the treatment of recursion.

The hierarchy of models in \cite{Ro10} roughly consist of two hierarchies: the stable models, and
the divergence-strict ones. The failures-divergences model could be seen as the centre piece in the
divergence-strict hierarchy, and the stable failures model \cite{Ro97}, which outside CSP stems
from {\sc Bergstra, Klop \& Olderog} \cite{BKO87}, plays the same role in the stable hierarchy.
Each of these hierarchies has a maximal (least discriminating) element, called $\cal{FL}^\Downarrow$
and $\cal{FL}$ in \cite{Ro10}. These correspond to the ready trace models ${\it RT}^\downarrow$ and
${\it RT}$ of \cite{vG93}.

The goal of the present paper is to propose a sensible model of CSP that is strictly finer than all
models thus far considered, and thus unites the two hierarchies mentioned above.
As all models of CSP considered so far have a distinctly linear time flavour, I here propose a
branching time model, thereby showing that the syntax of CSP is not predisposed towards linear time models.
My model can be given as an equivalence relation on labelled transition system, provided I show
that it is a congruence for the operators of CSP\@. I aim for an equivalence that allows a
complete axiomatisation in the style of \tab{axsFD}, obtained by replacing axioms that are no longer valid by
weaker ones.

One choice could be to base a model on strong bisimulation equivalence \cite{Mi90ccs}.
Strong bisimilarity is a congruence for all CSP operators, because their operational semantics fits
the tyft/tyxt format of \cite{GrV92}. However, this is an unsuitable equivalence for CSP, because it
fails to abstract from internal actions. Even the axiom $\Ii$ would not be valid, as the two sides
differ by an internal action.

A second proposal could be based on weak bisimilarity \cite{Mi90ccs}. This equivalence abstracts
from internal activity, and validates $\Ii$.
The default incarnation of weak bisimilarity is not finer than failures-divergences equivalence,
because it satisfies $\div\mathbin=\STOP$. Therefore, one would take a divergence-preserving variant of this
notion: the \emph{weak bisimulation with explicit divergence} of {\sc Bergstra, Klop \& Olderog} \cite{BKO87}.
Yet, some crucial CSP laws are invalidated, such as $\Ia$ and $\EI$. This destroys any hope of
a complete axiomatisation along the lines of \tab{axsFD}.

My final choice is \emph{divergence-preserving coupled similarity} \cite{vG93}, based on coupled
similarity for divergence-free processes from {\sc Parrow \& Sj\"odin} \cite{PS92}. This is the finest equivalence in
\cite{vG93} that satisfies $\Ia$ and $\EI$. In fact, it satisfies all of the axioms of \tab{axsFD},
except for the ones marked red: $\Bot$, $\Io$, $\Ei$, $\Ed$, \mbox{\textbf{\textbf{D2--4}}},
$\Sl$, $\PI$, $\PD$, $\HA$, $\UE$, $\UA$ and $\UD$.

Divergence-preserving coupled similarity belongs to the bisimulation family of semantic equivalences,
in the sense that on transition systems without internal actions it coincides with strong bisimilarity.

In \Sec{coupled} I present divergence-preserving coupled similarity.
In \Sec{congruence} I prove that it is a congruence for the operators of CSP,
and in \Sec{complete} I present a complete axiomatisation for recursion-free CSP processes without interrupts.

\section{Divergence-Preserving Coupled Similarity}\label{sec:coupled}

\begin{definition}\rm\label{df:coupled}
A \emph{coupled simulation} is a binary relation $\R$ on CSP processes,
such that, for all $\alpha\in\Sigma \cup\{\tau\}$,
\begin{itemize}
\item if $P \R Q$ and $P \goto{\alpha} P'$ then there exists a $Q'$ with $Q \dto{\hat\alpha} Q'$ and $P' \R Q'$,
\item and if $P \R Q$ then there exists a $Q'$ with $Q \dto{} Q'$ and $Q' \R P$.
\end{itemize}
It is \emph{divergence-preserving} if $P \R Q$ and $P {\Uparrow}$ implies $Q {\Uparrow}$.
Write \plat{$P \sqsupseteq_{CS}^\Delta Q$} if there exists a divergence-preserving coupled simulation $\R$ with $P \R Q$.
Two processes $P$ and $Q$ are \emph{divergence-preserving coupled similar}, notation \plat{$P \equiv_{CS}^\Delta Q$},
if $P \sqsupseteq_{CS}^\Delta Q$ and $Q \sqsupseteq_{CS}^\Delta P$.
\end{definition}
Note that the union of any collection of divergence-preserving coupled simulations is itself a
divergence-preserving coupled simulation. In particular, $\sqsupseteq_{CS}^\Delta$ is a
divergence-preserving coupled simulation.
Also note that in the absence of the internal action $\tau$, coupled simulations are symmetric, and
coupled similarity coincides with strong bisimilarity (as defined in \cite{Mi90ccs}).

Intuitively, \plat{$P \sqsupseteq_{CS}^\Delta Q$} says that $P$ is ``ahead'' of a state matching
$Q$, where $P'$ is ahead of $P$ if $P \dto{} P'$. The first clause says that if $P$ is ahead of a
state matching $Q$, then any transition performed by $P$ can be matched by $Q$---possibly after $Q$
``caught up'' with $P$ by performing some internal transitions. The second clause says that if $P$
is ahead of $Q$, then $Q$ can always catch up, so that it is ahead of $P$. Thus, if $P$ and $Q$ are
in stable states---where no internal actions are possible---then $P \sqsubseteq_{CS}^\Delta Q$
implies $Q \sqsubseteq_{CS}^\Delta P$.  In all other situations, $P$ and $Q$ do not need to be
matched exactly, but there do exists under- and overapproximations of a match. The result is that the
relation behaves like a weak bisimulation w.r.t.\ visible actions, but is not so pedantic in
matching internal actions.

\begin{proposition}\label{pr:preorder}
$\sqsupseteq_{CS}^\Delta$ is reflexive and transitive, and thus a preorder.
\end{proposition}
\begin{proof}
The identity relation $\Id$ is a divergence-preserving coupled simulation,
and if $\R$, $\R'$ are divergence-preserving coupled simulations, then so is $\R;\R' \cup \R';\R$.
Here ${\R};{\R'}$ is defined by $P \mathrel{{\R};{\R'}} R$ iff there is a $Q$ with $P \R Q \R' R$.

$\R{;}\R'$ is divergence-preserving: if $P {\R} Q {\R'} R$ and $P {\Uparrow}$, then $Q {\Uparrow}$,
and thus $R {\Uparrow}\!$.
The same holds for $\R';\R$, and thus for $\R;\R' \cup \R';\R$.

To check that $\R;\R' \cup \R';\R$ satisfies the first clause of \df{coupled},
note that if $Q \R' R$ and $Q \dto{\hat\alpha} Q'$, then,
by repeated application of the first clause of \df{coupled}, there is an $R'$ with $R \dto{\hat\alpha} R'$ and $Q' \R' R'$.

Towards the second clause, if $P \R Q \R' R$,
then, using the second clause for $\R$, there is a $Q'$ with $Q \dto{} Q'$ and $Q' \R P$.
Hence, using the first clause for $\R'$, there is an $R'$ with $R \dto{} R'$ and $Q' \R' R'$.
Thus, using the second clause for $\R'$, there is an $R''$ with $R' \dto{} R''$ and $R'' \R' Q'$,
and hence $R'' \mathrel{{\R'};{\R}} P'$.
\qed
\end{proof}

\begin{proposition}\label{pr:tau}
If $P \dto{} Q$ then $P \sqsubseteq_{CS}^\Delta Q$.
\end{proposition}
\begin{proof}
I show that $\Id \cup\{(Q,P)\}$, with $\Id$ the identity relation, is a coupled simulation.
Namely if \plat{$Q \goto{\alpha}Q'$} then surely \plat{$P \dto{\alpha} Q'$}.
The second clause of \df{coupled} is satisfied because \plat{$P \dto{} Q$}.
Furthermore, if $Q {\Uparrow}$ then certainly $P {\Uparrow}$, so the relation is divergence-preserving.
\qed
\end{proof}

\begin{proposition}\label{pr:CSpreorder}
$P \sqsupseteq_{CS}^\Delta Q$ iff $P\intchoice Q \equiv_{CS}^\Delta Q$.
\end{proposition}
\begin{proof}
``$\Rightarrow$'':
Let $\R$ be the smallest relation such that, for any $P$ and $Q$,
$P \sqsupseteq_{CS}^\Delta Q$ implies $P \R Q$, $(P\intchoice Q) \R Q$ and $Q \R (P\intchoice Q)$.
It suffices to show that $\R$ is a divergence-preserving coupled simulation.

That $\R$ is divergence-preserving is trivial, using that $(P \intchoice Q){\Uparrow}$ iff
$P{\Uparrow} \vee Q{\Uparrow}$.

Suppose $P^* \R Q$ and $P^* \goto{\alpha} P'$.
The case that $P^*=P$ with $P \sqsupseteq_{CS}^\Delta Q$ is trivial.
Now let $Q$ be $Q^* \intchoice P^*$. Since \plat{$P^* \goto{\alpha} P'$}, surely \plat{$Q \dto{\alpha} P'$}, and $P' \R P'$.
Finally, let $P^*=(P\intchoice Q)$ with \plat{$P \sqsupseteq_{CS}^\Delta Q$}.
Then $\alpha=\tau$ and $P'$ is either $P$ or $Q$. Both cases are trivial, taking $Q'=Q$.

Towards the second clause of \df{coupled}, suppose $P^* \R Q$.
The case $P^*=P$ with \plat{$P \sqsupseteq_{CS}^\Delta Q$} is trivial.
Now let $Q$ be $Q^* \intchoice P^*$. Then $Q\dto{} P^*$ and $P^* \R P^*$.
Finally, let $P^*=(P\intchoice Q)$ with \plat{$P \sqsupseteq_{CS}^\Delta Q$}.
Then $Q \dto{} Q$ and $Q \R (P\intchoice Q)$.

``$\Leftarrow$'':
Suppose \plat{$P \intchoice Q \sqsupseteq_{CS}^\Delta Q$}.
Since $P \intchoice Q \goto{\tau} P$ there exists a $Q'$ with $Q \dto{} Q'$ and $P \sqsupseteq_{CS}^\Delta Q'$.
By \pr{tau} \plat{$Q' \sqsupseteq_{CS}^\Delta Q$} and by \pr{preorder} $P \sqsupseteq_{CS}^\Delta Q$.
\qed
\end{proof}

\section{Congruence Properties}\label{sec:congruence}

\begin{proposition}\label{pr:prefixing}
$\equiv_{CS}^\Delta$ is a congruence for action prefixing.
\end{proposition}
\begin{proof}
I have to show that $P \equiv_{CS}^\Delta Q$ implies $(a\rightarrow P) \equiv_{CS}^\Delta (a\rightarrow Q)$.

Let $\R$ be the smallest relation such that, for any $P$ and $Q$,
$P \sqsubseteq_{CS}^\Delta Q$ implies $P \R Q$, and
$P \equiv_{CS}^\Delta Q$ implies $(a\rightarrow P) \R (a\rightarrow Q)$.
It suffices to show that $\R$ is a divergence-preserving coupled simulation.

Checking the conditions of \df{coupled} for the case $P \R Q$ with $P \sqsubseteq_{CS}^\Delta Q$ is trivial.
So I examine the case $(a\rightarrow P) \R (a\rightarrow Q)$ with $P \equiv_{CS}^\Delta Q$.

Suppose $(a\rightarrow P) \goto{\alpha} P'$. Then $\alpha=a$ and $P'=P$.
Now $(a\rightarrow Q)\goto{\alpha} Q$ and $P \R Q$, so the first condition of \df{coupled} is
satisfied.

For the second condition, $(a\rightarrow Q)\dto{} (a\rightarrow Q)$, and, since $Q \equiv_{CS}^\Delta P$,
$(a\rightarrow Q) \R (a\rightarrow P)$.
Thus, $\R$ is a coupled simulation.

As $a\rightarrow P$ does not diverge, $\R$ moreover is divergence-preserving.
\qed
\end{proof}
Since $\STOP \sqsupseteq_{CS}^\Delta (a\mathbin\rightarrow STOP)\timeout\STOP$ but
$\STOP \not\sqsubseteq_{CS}^\Delta (a\mathbin\rightarrow STOP)\timeout\STOP$,\linebreak and thus
\plat{$b\rightarrow \STOP \not\sqsupseteq_{CS}^\Delta b\rightarrow\big((a\rightarrow STOP) \timeout\STOP\big)$},
the relation \plat{$\sqsupseteq_{CS}^\Delta$} is \emph{not} a precongruence for action prefixing.

It is possible to express action prefixing in terms of the throw operator:
$a \rightarrow P$ is strongly bisimilar with $(a \rightarrow \STOP) \mathbin{\Theta_{\{a\}}} P$.
Consequently, \plat{$\sqsupseteq_{CS}^\Delta$} is not a precongruence for the throw operator.

\begin{proposition}
  $\equiv_{CS}^\Delta$ is a congruence for the throw operator.
\end{proposition}
\begin{proof}
  Let $A\subseteq \Sigma$. Let $\R$ be the smallest relation such that, for any $P_1,P_2,Q_1$, $Q_2$,
  $P_1 \sqsupseteq_{CS}^\Delta Q_1$ and $P_2 \equiv_{CS}^\Delta Q_2$ implies $P_1 \R Q_1$ and $(P_1\throw P_2) \R (Q_1\throw Q_2)$.
  It suffices to show that $\R$ is a divergence-preserving coupled simulation.

  So let $P_1 \sqsupseteq_{CS}^\Delta Q_1$, $P_2 \equiv_{CS}^\Delta Q_2$ and $(P_1\throw P_2) \goto{\alpha} P'$.
  Then \plat{$P_1 \goto{\alpha} P_1'$} for some $P_2'$, and
  either $\alpha\notin A$ and $P'=P_1' \throw P_2$, or $\alpha\in A$ and $P'=P_2$.
  So there is a $Q_1'$ with \plat{$Q_1 \dto{\hat\alpha} Q_1'$} and $P_1' \sqsupseteq_{CS}^\Delta Q_1'$.
  If $\alpha\notin A$ it follows that \plat{$(Q_1\throw Q_2) \dto{\hat\alpha} (Q_1'\throw Q_2)$}
  and $(P'_1\throw P_2) \R (Q'_1\throw Q_2)$.
  If $\alpha\in A$ it follows that \plat{$(Q_1\throw Q_2) \dto{\alpha} Q_2$} and $P_2 \R Q_2$.

  Now let $P_1 \sqsupseteq_{CS}^\Delta Q_1$ and $P_2 \equiv_{CS}^\Delta Q_2$.
  Then there is a $Q_1'$ with $Q_1 \dto{} Q_1'$ and \plat{$Q_1' \sqsupseteq_{CS}^\Delta P_1$}.
  Hence \plat{$Q_1\throw Q_2 \dto{} Q_1' \throw Q_2$} and $(Q_1'\throw Q_2) \R (P_1\throw P_2)$.

  The same two conditions for the case $P \R Q$ because $P \sqsupseteq_{CS}^\Delta Q$ are trivial.
  Thus $\R$ is a coupled simulation. That $\R$ is divergence-preserving follows because
  $P_1 \throw P_2 {\Uparrow}$ iff $P_1 {\Uparrow}$. 
\qed
\end{proof}
I proceed to show that \plat{$\sqsupseteq_{CS}^\Delta$} is a precongruence for all the other operators of
CSP\@. This implies that \plat{$\equiv_{CS}^\Delta$} is a congruence for all the operators of CSP\@.

\begin{proposition}
  $\sqsupseteq_{CS}^\Delta$ is a precongruence for internal choice.
\end{proposition}
\begin{proof}
  Let $\R$ be the smallest relation such that, for any $P_i$ and $Q_i$,
  $P_i \sqsupseteq_{CS}^\Delta Q_i$ for $i=1,2$ implies $P_i \R Q_i$ ($i=1,2$) and $(P_1\intchoice P_2) \R (Q_1\intchoice Q_2)$.
  It suffices to show that $\R$ is a divergence-preserving coupled simulation.

  So let \plat{$P_i \sqsupseteq_{CS}^\Delta Q_i$} for $i=1,2$ and \plat{$(P_1\intchoice P_2) \goto{\alpha} P'$}.
  Then $\alpha=\tau$ and $P'=P_i$ for $i=1$ or $2$.
  Now $Q_1 \intchoice Q_2 \dto{} Q_i$ and $P_i \R Q_i$.

  Now let \plat{$P_i \sqsupseteq_{CS}^\Delta Q_i$} for $i=1,2$.
  Then there is a $Q_1'$ with $Q_1 \dto{} Q_1'$ and \plat{$Q_1' \sqsupseteq_{CS}^\Delta P_1$}.
  By \pr{tau} \plat{$P_1 \sqsupseteq_{CS}^\Delta P_1\intchoice P_2$} and by \pr{preorder}
  $Q_1' \sqsupseteq_{CS}^\Delta P_1\intchoice P_2$.\vspace{1pt}

  The same two conditions for the case $P \R Q$ because $P \sqsupseteq_{CS}^\Delta Q$ are trivial.
  Thus $\R$ is a coupled simulation. That $\R$ is divergence-preserving follows because
  $P_1 \intchoice P_2 {\Uparrow}$ iff $P_1 {\Uparrow} \vee P_2 {\Uparrow}$.
  \qed
\pagebreak[2]
\end{proof}

\begin{proposition}
  $\sqsupseteq_{CS}^\Delta$ is a precongruence for external choice.
\end{proposition}
\begin{proof}
  Let $\R$ be the smallest relation such that, for any $P_i$ and $Q_i$,
  $P_i \sqsupseteq_{CS}^\Delta Q_i$ for $i=1,2$ implies $P_i \R Q_i$ ($i=1,2$) and $(P_1\extchoice P_2) \R (Q_1\extchoice Q_2)$.
  It suffices to show that $\R$ is a divergence-preserving coupled simulation.

  So let $P_i \sqsupseteq_{CS}^\Delta Q_i$ for $i=1,2$ and $(P_1\extchoice P_2) \goto{\alpha} P'$.
  If $\alpha\in\Sigma$ then $P_i \goto{\alpha} P'$ for $i=1$ or $2$, and there exists a $Q'$ with
  \plat{$Q_i \dto{\alpha} Q'$} and \plat{$P'\sqsupseteq_{CS}^\Delta Q'$}. Hence
  $Q_1\extchoice Q_2 \dto{\alpha} Q'$ and $P'\R Q'$.
  If $\alpha=\tau$ then either \plat{$P_1 \goto{\tau} P_1'$} for some $P_1'$ with $P'=P_1'\extchoice P_2$,
  or \plat{$P_2 \goto{\tau} P_2'$} for some $P_2'$ with $P'=P_1\extchoice P_2'$.
  I pursue only the first case, as the other follows by symmetry.
  Here \plat{$Q_1 \dto{} Q_1'$} for some $Q_1'$ with \plat{$P_1' \sqsupseteq_{CS}^\Delta Q_1'$}.
  Thus \plat{$Q_1\extchoice Q_2 \dto{} Q_1' \extchoice Q_2$} and $(P_1'\extchoice P_2) \R (Q_1'\extchoice Q_2)$.

  Now let \plat{$P_i \sqsupseteq_{CS}^\Delta Q_i$} for $i=1,2$.
  Then, for $i=1,2$, there is a $Q_i'$ with $Q_i \dto{} Q_i'$ and \plat{$Q_i' \sqsupseteq_{CS}^\Delta P_i$}.
  Hence \plat{$Q_1\extchoice Q_2 \dto{} Q_1' \extchoice Q_2'$} and $(Q_1'\extchoice Q_2') \R (P_1\extchoice P_2)$.

  Thus $\R$ is a coupled simulation. That $\R$ is divergence-preserving follows because
  $P_1 \extchoice P_2 {\Uparrow}$ iff $P_1 {\Uparrow} \vee P_2 {\Uparrow}$.
  \qed
\end{proof}

\begin{proposition}
  $\sqsupseteq_{CS}^\Delta$ is a precongruence for sliding choice.
\end{proposition}
\begin{proof}
  Let $\R$ be the smallest relation such that, for any $P_i$ and $Q_i$,
  $P_i \sqsupseteq_{CS}^\Delta Q_i$ for $i=1,2$ implies $P_i \R Q_i$ ($i=1,2$) and $(P_1\timeout P_2) \R (Q_1\timeout Q_2)$.
  It suffices to show that $\R$ is a divergence-preserving coupled simulation.

  So let $P_i \sqsupseteq_{CS}^\Delta Q_i$ for $i=1,2$ and $(P_1\timeout P_2) \goto{\alpha} P'$.
  If $\alpha\in\Sigma$ then $P_1 \goto{\alpha} P'$, and there exists a $Q'$ with
  \plat{$Q_1 \dto{\alpha} Q'$} and \plat{$P'\sqsupseteq_{CS}^\Delta Q'$}. Hence
  \plat{$Q_1\timeout Q_2 \dto{\alpha} Q'$} and $P'\R Q'$.
  If $\alpha\mathbin=\tau$ then either $P'\mathbin=P_2$ or \plat{$P_1 \goto{\tau} P_1'$}
  for some $P_1'$ with $P'\mathbin=P_1'\timeout P_2$.
  In the former case \plat{$Q_1 \timeout Q_2 \dto{} Q_2$} and $P_2 \R Q_2$.
  In the latter case \plat{$Q_1 \dto{} Q_1'$} for some $Q_1'$ with \plat{$P_1' \sqsupseteq_{CS}^\Delta Q_1'$}.
  Thus \plat{$Q_1\timeout Q_2 \dto{} Q_1' \timeout Q_2$} and $(P_1'\timeout P_2) \R (Q_1'\timeout Q_2)$.

  Now let \plat{$P_i \sqsupseteq_{CS}^\Delta Q_i$} for $i=1,2$.
  Then there is a $Q_2'$ with $Q_2 \dto{} Q_2'$ and \plat{$Q_2' \sqsupseteq_{CS}^\Delta P_2$}.
  By \pr{tau} \plat{$P_2 \sqsupseteq_{CS}^\Delta P_1\timeout P_2$} and by \pr{preorder}
  $Q_2' \sqsupseteq_{CS}^\Delta P_1\timeout P_2$.

  Thus $\R$ is a coupled simulation. That $\R$ is divergence-preserving follows because
  $P_1 \timeout P_2 {\Uparrow}$ iff $P_1 {\Uparrow} \vee P_2 {\Uparrow}$.
  \qed
\end{proof}

\begin{proposition}
  $\sqsupseteq_{CS}^\Delta$ is a precongruence for parallel composition.
\end{proposition}
\begin{proof}
  Let $A\subseteq \Sigma$.
  Let $\R$ be the smallest relation such that, for any $P_i$ and $Q_i$,
  $P_i \sqsupseteq_{CS}^\Delta Q_i$ for $i=1,2$ implies $(P_1\|_A P_2) \R (Q_1\|_A Q_2)$.
  It suffices to show that $\R$ is a divergence-preserving coupled simulation.

  So let $P_i \sqsupseteq_{CS}^\Delta Q_i$ for $i=1,2$ and $(P_1\|_A P_2) \goto{\alpha} P'$.
  If $\alpha\notin A$ then $P_i \goto{\alpha} P_i'$ for $i=1$ or $2$, and $P'=P_1'\|_A P_2'$,
  where $P_{3-i}':=P_{3-i}$. Hence there exists a $Q_i'$ with
  \plat{$Q_i \dto{\hat\alpha} Q_i'$} and \plat{$P_i'\sqsupseteq_{CS}^\Delta Q_i'$}. Let $Q_{3-i}':=Q_{3-i}$.
  Then \plat{$Q_1\|_A Q_2 \dto{\hat\alpha} Q_1'\|Q'_2$} and $(P_1'\|P'_2)\R (Q_1'\|Q'_2)$.
  If $\alpha\in A$ then \plat{$P_i \goto{\alpha} P_i'$} for $i=1$ and $2$.
  Hence, for $i=1,2$, \plat{$Q_i \dto{\alpha} Q_i'$} for some $Q_i'$ with \plat{$P_i' \sqsupseteq_{CS}^\Delta Q_i'$}.
  Thus \plat{$Q_1\|_A Q_2 \dto{\alpha} Q_1' \|_A Q_2'$} and $(P_1'\|_A P'_2) \R (Q_1'\|_A Q'_2)$.

  Now let \plat{$P_i \sqsupseteq_{CS}^\Delta Q_i$} for $i=1,2$.
  Then, for $i=1,2$, there is a $Q_i'$ with $Q_i \dto{} Q_i'$ and \plat{$Q_i' \sqsupseteq_{CS}^\Delta P_i$}.
  Hence \plat{$Q_1\|_A Q_2 \dto{} Q_1' \|_A Q_2'$} and $(Q_1'\|_A Q_2') \R (P_1\|_A P_2)$.

  Thus $\R$ is a coupled simulation. That $\R$ is divergence-preserving follows because
  $P_1 \|_A P_2 {\Uparrow}$ iff $P_1 {\Uparrow} \vee P_2 {\Uparrow}$.
  \qed
\end{proof}

\begin{proposition}
  $\sqsupseteq_{CS}^\Delta$ is a precongruence for concealment.
\end{proposition}
\begin{proof}
Let $A\subseteq \Sigma$. Let $\R$ be the smallest relation such that, for any $P$ and $Q$,
$P \sqsubseteq_{CS}^\Delta Q$ implies $(P\conceal A) \R (Q \conceal A)$.
It suffices to show that $\R$ is a divergence-preserving coupled simulation.

So let $P \sqsubseteq_{CS}^\Delta Q$ and $P\conceal A \goto{\alpha} P^*$.
Then $P^*=P'\conceal A$ for some $P'$ with \plat{$P \goto{\beta} P'$},
and either $\beta\in A$ and $\alpha=\tau$, or $\beta=\alpha\notin A$.
Hence \plat{$Q \goto{\beta} Q'$} for some $Q'$ with $P' \sqsubseteq_{CS}^\Delta Q'$.
Therefore \plat{$Q\conceal A \goto{\alpha} Q'\conceal A$} and $(P'\conceal A) \R (Q'\conceal A)$.

Now let $P \sqsubseteq_{CS}^\Delta Q$. Then there is a $Q'$ with $Q \dto{} Q'$ and \plat{$Q' \sqsupseteq_{CS}^\Delta P$}.
Hence \plat{$Q\conceal A \dto{} Q'\conceal A$} and \plat{$(Q'\conceal A) \R (P\conceal A)$}.

To check that $\R$ is divergence-preserving, suppose $(P\conceal A) {\Uparrow}$.
Then there are $P_i$ and $\alpha_i \in A\cup\{\tau\}$ for all $i>0$ such that
\plat{$P \goto{\alpha_1} P_1 \goto{\alpha_2} P_2 \goto{\alpha_3} \dots$}.
By the first condition of \df{coupled}, there are $Q_i$ for all $i>0$ such that
$P_i \R Q_i$ and \plat{$Q \dto{\hat\alpha_1} Q_1 \dto{\hat\alpha_2} Q_2 \dto{\hat\alpha_3} \dots$}.
This implies $Q\conceal A \dto{} Q_1\conceal A \dto{} Q_2\conceal A \dto{} \dots$.

In case $\alpha_i \in \Sigma$ for infinitely many $i$, then for infinitely many $i$ one has
$Q_{i-1} \dto{\alpha_i} Q_i$ and thus \plat{$Q_{i-1}\conceal A \dto{\tau} Q_i\conceal A$}.
This implies that $(Q\conceal A) {\Uparrow}$.

Otherwise there is an $n>0$ such that $\alpha_i=\tau$ for all $i \geq n$.
In that case $P_n {\Uparrow}$ and thus $Q_n {\Uparrow}$.
Hence $(Q_n\conceal A) {\Uparrow}$ and thus $(Q\conceal A) {\Uparrow}$.
\qed
\end{proof}

\begin{proposition}
  $\sqsupseteq_{CS}^\Delta$ is a precongruence for renaming.
\end{proposition}
\begin{proof}
Let $f:\Sigma\rightarrow \Sigma$.
Let $\R$ be the smallest relation such that, for any $P$ and $Q$,
$P \sqsubseteq_{CS}^\Delta Q$ implies $f(P) \R f(Q)$.
It suffices to show that $\R$ is a divergence-preserving coupled simulation.

So let $P \sqsubseteq_{CS}^\Delta Q$ and $f(P) \goto{\alpha} P^*$.
Then $P^*=f(P')$ for some $P'$ with \plat{$P \goto{\beta} P'$} and $f(\beta)=\alpha$.
Hence \plat{$Q \goto{\beta} Q'$} for some $Q'$ with $P' \sqsubseteq_{CS}^\Delta Q'$.
Therefore \plat{$f(Q) \goto{\alpha} f(Q')$} and $f(P') \R f(Q')$.

Now let $P \sqsubseteq_{CS}^\Delta Q$. Then there is a $Q'$ with $Q \dto{} Q'$ and \plat{$Q' \sqsupseteq_{CS}^\Delta P$}.
Hence \plat{$f(Q) \dto{} f(Q')$} and \plat{$f(Q') \R f(P)$}.

To check that $\R$ is divergence-preserving, suppose $f(P) {\Uparrow}$.
Then $P {\Uparrow}$, so $Q {\Uparrow}$ and $f(Q) {\Uparrow}$.
\qed
\end{proof}

\begin{proposition}
  $\sqsupseteq_{CS}^\Delta$ is a precongruence for the interrupt operator.
\end{proposition}
\begin{proof}
  Let $\R$ be the smallest relation such that, for any $P_i$ and $Q_i$,
  $P_i \sqsupseteq_{CS}^\Delta Q_i$ for $i=1,2$ implies $P_2 \R Q_2$ and $(P_1\interrupt P_2) \R (Q_1\interrupt Q_2)$.
  It suffices to show that $\R$ is a divergence-preserving coupled simulation.

  So let $P_i \sqsupseteq_{CS}^\Delta Q_i$ for $i=1,2$ and $(P_1\interrupt P_2) \goto{\alpha} P'$.
  Then either $P'=P_1' \interrupt P_2$ for some $P_1'$ with \plat{$P_1 \goto{\alpha} P_1'$},
  or $\alpha=\tau$ and $P'=P_1 \interrupt P_2'$ for some $P_2'$ with \plat{$P_2 \goto{\tau} P_2'$},
  or $\alpha\in\Sigma$ and \plat{$P_2  \goto{\alpha} P'$}.

  In the first case there is a $Q_1'$ with \plat{$Q_1 \dto{\hat\alpha} Q_1'$} and $P_1' \sqsupseteq_{CS}^\Delta Q_1'$.
  It follows that \plat{$(Q_1\interrupt Q_2) \dto{\hat\alpha} (Q_1'\interrupt Q_2)$}
  and $(P'_1\interrupt P_2) \R (Q'_1\interrupt Q_2)$.

  In the second case there is a $Q_2'$ with \plat{$Q_2 \dto{} Q_2'$} and $P_2' \sqsupseteq_{CS}^\Delta Q_2'$.
  It follows that \plat{$(Q_1\interrupt Q_2) \dto{} (Q_1\interrupt Q_2')$}
  and $(P_1\interrupt P_2') \R (Q_1\interrupt Q_2')$.

  In the last case there is a $Q_2'$ with \plat{$Q_2 \dto{\alpha} Q_2'$} and $P_2' \sqsupseteq_{CS}^\Delta Q_2'$.
  It follows that \plat{$(Q_1\interrupt Q_2) \dto{\alpha}  Q_2'$} and $P_2' \R Q_2'$.

  Now let \plat{$P_i \sqsupseteq_{CS}^\Delta Q_i$} for $i=1,2$.
  Then, for $i=1,2$, there is a $Q_i'$ with $Q_i \dto{} Q_i'$ and \plat{$Q_i' \sqsupseteq_{CS}^\Delta P_i$}.
  Hence \plat{$Q_1\interrupt Q_2 \dto{} Q_1' \interrupt Q_2'$} and $(Q_1'\interrupt Q_2') \R (P_1\interrupt P_2)$.

  Thus $\R$ is a coupled simulation. That $\R$ is divergence-preserving follows because
  $P_1 \interrupt P_2 {\Uparrow}$ iff $P_1 {\Uparrow} \vee P_2 {\Uparrow}$.
\qed
\end{proof}

\section{A Complete Axiomatisation of $\equiv_{CS}^\Delta$}\label{sec:complete}

A set of equational laws valid for $\equiv_{CS}^\Delta$ is presented in \tab{axsCS}.
It includes the laws from \tab{axsFD} that are still valid for $\equiv_{CS}^\Delta$.
I will show that this axiomatisation is sound and complete for $\equiv_{CS}^\Delta$ for
recursion-free CSP without the interrupt operator. The axioms $\UE$ and $\UA$, which are not valid
for $\equiv_{CS}^\Delta$, played a crucial r\^ole in reducing CSP expressions with interrupt into normal form.
It is not trivial to find valid replacements, and due to lack of space and time I do not tackle this problem here.

The axiom $\CE$ replaces the fallen axiom $\HA$, and is due to \cite{Ro10}.
Here the result of hiding actions results in a process that cannot be expressed as a normal form built
up from $a \rightarrow$, $\intchoice$ and $\extchoice$. For this reason, one needs a richer normal
form, involving the sliding choice operator. It is given by the context-free grammar
\[\begin{array}{l}
N \rightarrow D \mid D \timeout I \\
I \rightarrow D \mid I \intchoice I \\
D \rightarrow \STOP \mid \div \mid E \mid \div\extchoice E \\
E \rightarrow (a\rightarrow N) \mid (a\rightarrow N)\extchoice E\;.
\vspace{-2ex}
\end{array}\]
\begin{definition}\rm
A CSP expression is in \emph{head normal form} if  it is of the form
$\big([\div \mathrel{\extchoice]} \Extchoice_{i\in I} (a_i\rightarrow R_i)\big)\timeout \Intchoice_{j\in J}R_j$,
with $R_j=\big([\div \mathrel{\extchoice]}\Extchoice_{k\in K_j}(a_{kj}\rightarrow R_{kj})\big)$ for $j\in J$.
Here $I$, $J$ and the $K_j$ are finite index sets, and the parts between square brackets are optional.
Here, although \plat{$\Intchoice_{i\in \emptyset}P_i$} is undefined, I use \plat{$P \timeout \Intchoice_{i\in \emptyset}P_i$}\vspace{2pt}
to represent $P$. An expression is in \emph{normal form} if it has this form and also
the subexpressions $R_i$ and $R_{kj}$ are in normal form.

A head normal form is \emph{saturated} if the $\div$-summand on the left is present
whenever any of the $R_j$ has a $\div$-summand, and
for any $j\mathbin\in J$ and any $k\mathbin\in K_j$ there is an $i\mathbin\in I$ 
with $a_i\mathbin=a_{kj}$ and $R_i \mathbin= R_{kj}$.
\end{definition}
My proof strategy is to ensure that there are enough axioms to transform any CSP process without recursion
and interrupt operators into normal form, and to make these forms saturated; then to
equate saturated normal forms that are divergence-preserving coupled simulation equivalent.

Due to the optional presence in head normal forms of a $\div$-summand and a sliding choice,
I need four variants of the axiom $\CE$; so far I have not seen a way around this.
Likewise, there are $4\times 4$ variants of the axiom $\PE$ from \tab{axsFD}, of which 6 could be
suppressed by symmetry (\textbf{P4--P13}). There are also 3 axioms replacing $\PI$ (\textbf{P14--P16}).

\begin{table}[p]
\[
\begin{array}{@{}lrcl@{}}
\Ii  &  P \intchoice P              & = & P                       \\
\Ic  &  P \intchoice Q              & = & Q \intchoice P              \\
\Ia  &  P \intchoice (Q \intchoice R)   & = & (P \intchoice Q) \intchoice R   \\[1ex]
\Ec  &  P \extchoice Q                & = & Q \extchoice P                \\
\Ea  &  P \extchoice (Q \extchoice R)       & = & (P \extchoice Q) \extchoice R       \\
\Es  &  P \extchoice \STOP            & = & P                       \\[1ex]
\Si  &  P \timeout P & = & P \\
\Sa  &  P \timeout (Q \timeout R)  & = & (P \timeout Q)\timeout R \\
\SE  &  (P \timeout Q)\timeout R   & = & (P \extchoice Q) \timeout R \\
\SI  &  (P \intchoice Q)\timeout R & = & (P \extchoice Q) \timeout R \\
\Ss  &  \STOP \timeout P            & = & P                       \\[1ex]
\IS  &  (P \timeout Q) \intchoice (R \timeout S) & = & (P \extchoice R) \timeout (Q\intchoice S)\\
\ES  &  (P \timeout Q) \extchoice (R \timeout S) & = & (P \extchoice R) \timeout (Q \extchoice S)\\
\EI  &  P \extchoice (Q \intchoice R)     & = & (P \extchoice Q) \intchoice (P \extchoice R) \\[1ex]
\Pru  &  (a\rightarrow P) \extchoice a\rightarrow(P \intchoice Q) & = & a\rightarrow(P \intchoice Q)\\[1ex]
\Pa  &  P \|_A (Q \|_A R)       & = & (P \|_A Q) \|_A R        \\
\Pc  &  P \|_A Q                & = & Q \|_A P                 \\
\textbf{P4--P13} & \multicolumn{3}{c}{\mbox{\it more axioms for parallel composition follow on the next page}}\\[1ex]
\HI  &  (P \intchoice Q)\conceal A   & = & (P\conceal A) \intchoice (Q\conceal A) \\
\CE  &  \big(\Extchoice_{i\in I} (a_i \rightarrow P_i )\big)\conceal A & = &
                 \big(\Extchoice_{a_i\not\in A} (a_i \rightarrow (P_i\conceal A) )\big)\\
              &&&\mbox{}  \timeout \Intchoice_{a_i\in A} (P_i\conceal A)\\
\CED  &  \big(\div \extchoice \Extchoice_{i\in I} (a_i \rightarrow P_i )\big)\conceal A & = &
                 \big(\div \extchoice \Extchoice_{a_i\not\in A} (a_i \rightarrow (P_i\conceal A) )\big)\\
              &&&\mbox{} \timeout \Intchoice_{a_i\in A} (P_i\conceal A)\\
\CEI  &  \left(\big(\Extchoice_{i\in I} (a_i \rightarrow P_i )\big)\timeout P'\right)\conceal A & = &
                 \big(\Extchoice_{a_i\not\in A} (a_i \rightarrow (P_i\conceal A) )\big)\\
              &&&\mbox{} \timeout \big( P'\conceal A \intchoice \Intchoice_{a_i\in A} (P_i\conceal A)\big)\\
\CEDI  &  \left(\big(\div \extchoice \Extchoice_{i\in I} (a_i \rightarrow P_i )\big)\timeout P'\right)\conceal A & = &
                 \big(\div \extchoice \Extchoice_{a_i\not\in A} (a_i \rightarrow (P_i\conceal A) )\big)\\
              &&&\mbox{} \timeout \big( P'\conceal A \intchoice \Intchoice_{a_i\in A} (P_i\conceal A)\big)\\[1ex]
\RS  &  f(P \timeout Q)         & = & f(P) \timeout f(Q)   \\
\RI  &  f(P \intchoice Q)         & = & f(P) \intchoice f(Q)   \\
\RE  &  f(P \extchoice Q)         & = & f(P) \extchoice f(Q)   \\
\RA  &  f(a \rightarrow P)        & = & f(a) \rightarrow f(P)   \\
\Rs  &  f(\STOP)                   & = & \STOP   \\
\RD  &  f(\div)                   & = & \div   \\[1ex]
\TS  &  (P \timeout Q)\throw R           & = & (P\throw R) \timeout (Q\throw R)   \\
\TI  &  (P \intchoice Q)\throw R         & = & (P\throw R) \intchoice (Q\throw R)   \\
\TE  &  (P \extchoice Q)\throw R         & = & (P\throw R) \extchoice (Q\throw R)   \\
\TA  &  (a \rightarrow P) \throw Q       & = & a \rightarrow (P\throw Q) \hfill\mbox{if $a \notin A$}  \\
\AT  &  (a \rightarrow P) \throw Q       & = & a \rightarrow Q \hfill \mbox{if $a \in A$}  \\
\Ts  &  \STOP \throw Q       & = & \STOP  \\
\TD  &  \div \throw Q       & = & \div  \\[1ex]
%
\end{array}
\]
\caption{A complete axiomatisation of $\equiv_{CS}^\Delta$ for recursion-free CSP without interrupt}\label{tab:axsCS}
\end{table}

\begin{table}
Below $P = \Extchoice_{i \in I}(a_i\rightarrow P_i)$
and $Q = \Extchoice_{j \in J}(b_j\rightarrow Q_j)$.\vspace{-5pt}
\[
\begin{array}{@{}l@{\quad}rcl@{}}
\mathbf{(P4)}  &  P \|_A Q                  & = &
                  \Extchoice_{a_i\notin A} (a_i \rightarrow (P_i\|_A Q)) \extchoice \\
             &&& \Extchoice_{a_j=b_j\in A} (a_i \rightarrow (P_i \|_A Q_j)) \extchoice\\
             &&& \Extchoice_{b_j\notin A} (b_j \rightarrow (P\|_A Q_j)) \\
\mathbf{(P5)}  &  (\div \extchoice P) \|_A Q                  & = &
                  \div \extchoice \Extchoice_{a_i\notin A} (a_i \rightarrow (P_i\|_A Q)) \extchoice \\
             &&& \Extchoice_{a_j=b_j\in A} (a_i \rightarrow (P_i \|_A Q_j)) \extchoice\\
             &&& \Extchoice_{b_j\notin A} (b_j \rightarrow ((\div \extchoice P)\|_A Q_j)) \\
\mathbf{(P6)}  &  (\div \extchoice P) \|_A (\div \extchoice Q)                  & = &
                  \div \extchoice \Extchoice_{a_i\notin A} (a_i \rightarrow (P_i\|_A (\div \extchoice Q))) \extchoice \\
             &&& \Extchoice_{a_j=b_j\in A} (a_i \rightarrow (P_i \|_A Q_j)) \extchoice\\
             &&& \Extchoice_{b_j\notin A} (b_j \rightarrow ((\div \extchoice P)\|_A Q_j)) \\
\mathbf{(P7)}  &  (P \timeout P') \|_A Q                  & = &
                  \big(\Extchoice_{a_i\notin A} (a_i \rightarrow (P_i\|_A Q))  \extchoice \\
             &&& \Extchoice_{a_j=b_j\in A} (a_i \rightarrow (P_i \|_A Q_j)) \extchoice\\
             &&& \Extchoice_{b_j\notin A} (b_j \rightarrow ((P\timeout P')\|_A Q_j))\big)
                 \timeout P' \|_A Q\\
\mathbf{(P8)}  &  ((\div \extchoice P) \timeout P') \|_A Q                  & = &
                 \big(\div \extchoice \Extchoice_{a_i\notin A} (a_i \rightarrow (P_i\|_A Q))  \extchoice \\
             &&& \Extchoice_{a_j=b_j\in A} (a_i \rightarrow (P_i \|_A Q_j)) \extchoice\\
             &&& \Extchoice_{b_j\notin A} (b_j \rightarrow (((\div \extchoice P)\timeout P')\|_A Q_j))\big)\\
             &&& \mbox{} \timeout P' \|_A Q\\
\mathbf{(P9)}  &  (P \timeout P') \|_A (\div \extchoice Q)                 & = &
                  \big((\div \extchoice \Extchoice_{a_i\notin A} (a_i \rightarrow (P_i\|_A (\div \extchoice Q) ))  \extchoice \\
             &&& \Extchoice_{a_j=b_j\in A} (a_i \rightarrow (P_i \|_A Q_j)) \extchoice\\
             &&& \Extchoice_{b_j\notin A} (b_j \rightarrow ((P\timeout P')\|_A Q_j))\big)\\
             &&& \mbox{}\timeout P' \|_A (\div \extchoice Q) \\
\mathbf{(P10)} &  ((\div \extchoice P) \timeout P') \|_A (\div \extchoice Q)                  & = &
                 \big(\div \extchoice \Extchoice_{a_i\notin A} (a_i \rightarrow (P_i\|_A (\div \extchoice Q)))  \extchoice \\
             &&& \Extchoice_{a_j=b_j\in A} (a_i \rightarrow (P_i \|_A Q_j)) \extchoice\\
             &&& \Extchoice_{b_j\notin A} (b_j \rightarrow (((\div \extchoice P)\timeout P')\|_A Q_j))\big)\\
             &&& \mbox{} \timeout P' \|_A (\div \extchoice Q)\\
\mathbf{(P11)} &  (P \timeout P') \|_A (Q \timeout Q')  & = &
                 \big(\Extchoice_{a_i\notin A} (a_i \rightarrow (P_i\|_A (Q\timeout Q')))  \extchoice \\
             &&& \Extchoice_{a_j=b_j\in A} (a_i \rightarrow (P_i \|_A Q_j)) \extchoice\\
             &&& \Extchoice_{b_j\notin A} (b_j \rightarrow ((P\timeout P')\|_A Q_j))\big)\\
             &&& \mbox{} \timeout \big(P' \|_A (Q\timeout Q')  \intchoice (P \timeout P') \|_A Q'\big)\\[3pt]
\mathbf{(P12)} &  ((\div \extchoice P) \timeout P') \|_A (Q \timeout Q')  & = &
                 \big(\div \extchoice \Extchoice_{a_i\notin A} (a_i \rightarrow (P_i\|_A (Q\timeout Q')))  \extchoice \\
             &&& \Extchoice_{a_j=b_j\in A} (a_i \rightarrow (P_i \|_A Q_j)) \extchoice\\
             &&& \Extchoice_{b_j\notin A} (b_j \rightarrow (((\div \extchoice P)\timeout P')\|_A Q_j))\big)\\
             &&& \mbox{}\timeout \big(P' \|_A (Q\timeout Q')  \intchoice ((\div \extchoice P) \timeout P') \|_A Q'\big)\\
\end{array}
\]
\begin{center}
\textbf{Table~\ref*{tab:axsCS}.} A complete axiomatisation of $\equiv_{CS}^\Delta$ for recursion-free CSP (continued)
\end{center}
\end{table}
\begin{table}[t]
\[
\begin{array}{@{}l@{\qquad\quad}rl@{}}
\mathbf{(P13)} &  ((\div \extchoice P) \timeout P') \|_A ((\div \extchoice Q) \timeout Q')  = \big(\div \extchoice
              \hspace{-100pt} \\
             && \Extchoice_{a_i\notin A} (a_i \rightarrow (P_i\|_A ((\div \extchoice Q)\timeout Q'))) \extchoice \\
             && \Extchoice_{a_j=b_j\in A} (a_i \rightarrow (P_i \|_A Q_j)) \extchoice\\
             && \Extchoice_{b_j\notin A} (b_j \rightarrow (((\div \extchoice P)\timeout P')\|_A Q_j))\big)\\
             && \mbox{}\timeout
                 \big(P' \|_A ((\div \extchoice Q)\timeout Q')  \intchoice ((\div \extchoice P) \timeout P') \|_A Q'\big)\\
\end{array}
\]
Below $P = \Extchoice_{i \in I}(a_i\rightarrow P_i)$
and $Q = \Intchoice_{j \in J}Q_j$.\vspace{-5pt}
\[
\begin{array}{@{}l@{\quad}rcl@{}}
\mathbf{(P14)}  &  P \|_A Q                  & = &
                  \Extchoice_{a_i\notin A} (a_i \rightarrow (P_i\|_A Q)) \timeout \Intchoice_{j\in J} (P\|_A Q_j)) \\
\mathbf{(P15)}  &  (\div \extchoice P )\|_A Q                  & = &
                  \big(\div \extchoice \Extchoice_{a_i\notin A} (a_i \rightarrow (P_i\|_A Q))\big) \timeout
                   \Intchoice_{j\in J} ((\div \extchoice P)\|_A Q_j)) \\
\end{array}
\]
Below $P = \Intchoice_{i \in I}P_i$
and $Q = \Intchoice_{j \in J}Q_j$.\vspace{-5pt}
\[
\begin{array}{@{}l@{\qquad\qquad\qquad\qquad\quad}rcl@{\qquad\qquad\qquad}}
\mathbf{(P16)}  &  P \|_A Q                  & = &
                  \Intchoice_{a_i\notin A} (P_i\|_A Q) \intchoice \Intchoice_{j\in J} (P\|_A Q_j)) \\
\end{array}
\]
\begin{center}
\textbf{Table~\ref*{tab:axsCS}.} A complete axiomatisation of $\equiv_{CS}^\Delta$ for recursion-free CSP (continued)
\end{center}
\vspace{-1ex}
\end{table}

\section{Soundness}

Since divergence-preserving coupled similarity is a congruence for all CSP operators, to establish
the soundness of the axiomatisation of \tab{axsCS} it suffices to show the validity w.r.t.\ $\equiv_{CS}^\Delta$ of all axioms.
When possible, I show validity w.r.t.\ strong bisimilarity, which is a strictly finer equivalence.

\begin{definition}\rm
Two processes are \emph{strongly bisimilar} \cite{Mi90ccs} if they are related by a binary relation
$\R$ on processes such that, for all $\alpha\in\Sigma \cup\{\tau\}$,
\begin{itemize}
\item if $P \R Q$ and $P \goto{\alpha} P'$ then there exists a $Q'$ with $Q \goto{\alpha} Q'$ and $P' \R Q'$,
\item if $P \R Q$ and $Q \goto{\alpha} Q'$ then there exists a $P'$ with $P \goto{\alpha} P'$ and $P' \R Q'$.
\end{itemize}
\end{definition}

\begin{proposition}
  Axiom $\Ii$ is valid for $\equiv_{CS}^\Delta$.
\end{proposition}
\begin{proof}
  $\{(P\intchoice P, P), (P, P\intchoice P) \mid P ~~\mbox{a process}\} \cup \Id$
  is a divergence-preserving coupled simulation.
\qed
\end{proof}

\begin{proposition}
  Axiom $\Ic$ is valid even for strong bisimilarity.
\end{proposition}
\begin{proof}
  $\{(P\intchoice Q, Q \intchoice P)\mid P,Q ~~\mbox{processes}\} \cup \Id$ is a strong bisimulation.
\qed
\end{proof}

\begin{proposition}
  Axiom $\Ia$ is valid for $\equiv_{CS}^\Delta$.
\end{proposition}
\begin{proof}
  The relation $\{\big( P \intchoice (Q \intchoice R) , (P \intchoice Q) \intchoice R \big),
  \big((P \intchoice Q) \intchoice R , P \intchoice (Q \intchoice R) \big),\linebreak[3]
  \big(Q \intchoice R , (P \intchoice Q) \intchoice R \big),
  \big(P \intchoice Q, P \intchoice (Q \intchoice R)\big), \big(R,Q \intchoice R\big),
  \big(P,P\intchoice Q\big) \,|\, P,Q,R ~~\mbox{processes}\}\linebreak \cup \Id$
  is a divergence-preserving coupled simulation.
\qed
\pagebreak[3]
\end{proof}

\begin{proposition}
  Axioms {\bf E2--4} are valid for strong bisimilarity.
\end{proposition}
\begin{proof}
  The relation $\{\big( P \extchoice (Q \extchoice R) , (P \extchoice Q) \extchoice R \big)
  \mid P,Q,R ~~\mbox{processes}\} \cup \Id$ is a strong bisimulation.
  So is $\{(P\extchoice Q, Q \extchoice P)\mid P,Q ~~\mbox{processes}\} \cup \Id$,
  as well as $\{(P\extchoice \STOP, P)\mid P ~~\mbox{a process}\} \cup \Id$.
\qed
\end{proof}

\begin{proposition}
  Axiom $\Si$ is valid for $\equiv_{CS}^\Delta$.
\end{proposition}
\begin{proof}
  $\{(P' \timeout P, P),(P,P' \timeout P) \mid P' \sqsupseteq_{CS}^\Delta P\} \cup \Id$
  is a divergence-preserving coupled simulation. This follows from \pr{tau}.
\qed
\end{proof}

\begin{proposition}
  Axiom $\Sa$ is valid for $\equiv_{CS}^\Delta$.
\end{proposition}
\begin{proof}
  $\{\big( P \timeout (Q \timeout R) , (P \timeout Q) \timeout R \big),
  \big((P \timeout Q) \timeout R , P \timeout (Q \timeout R) \big) \mid P,Q,R ~~\mbox{processes}\} \cup \Id$
  is a divergence-preserving coupled simulation.
\qed
\pagebreak[2]
\end{proof}

\begin{proposition}
  Axiom $\SE$ is valid for $\equiv_{CS}^\Delta$.
\end{proposition}
\begin{proof}
  $\{\big( (P \timeout Q) \timeout R , (P \extchoice Q) \timeout R \big),
  \big( (P \extchoice Q') \timeout R , (P \timeout Q) \timeout R \big),
  \big( Q \timeout R , (P \extchoice Q) \timeout R \big),\linebreak[3]
  \big( R , Q \timeout R \big)
  \mid  Q' \sqsupseteq_{CS}^\Delta Q\} \cup \Id$
  is a divergence-preserving coupled simulation.
\qed
\end{proof}

\begin{proposition}
  Axiom $\SI$ is valid for $\equiv_{CS}^\Delta$.
\end{proposition}
\begin{proof}
  $\{\big( (P \intchoice Q) \timeout R , (P \extchoice Q) \timeout R \big),
  \big( (P' \extchoice Q') \timeout R , (P \intchoice Q) \timeout R \big),
  \big( P \timeout R , (P \extchoice Q) \timeout R \big),\linebreak[3]
  \big( Q \timeout R , (P \extchoice Q) \timeout R \big),
  \big( R , Q \timeout R \big)
  \mid P' \sqsupseteq_{CS}^\Delta P \wedge Q' \sqsupseteq_{CS}^\Delta Q\} \cup \Id$
  is a divergence-preserving coupled simulation. Checking this involves \pr{tau}.
\qed
\end{proof}

\begin{proposition}
  Axiom $\Ss$ is valid for $\equiv_{CS}^\Delta$.
\end{proposition}
\begin{proof}
  The relation $\{(\STOP\timeout P, P), (P, \STOP\timeout P) \mid P ~~\mbox{a process}\} \cup \Id$
  is a divergence-preserving coupled simulation.
\qed
\end{proof}

\begin{proposition}
  Axiom $\IS$ is valid for $\equiv_{CS}^\Delta$.
\end{proposition}
\begin{proof}
  $\{\big( (P \timeout Q) \intchoice (R \timeout S) , (P \extchoice R) \timeout (Q\intchoice S)\big),
     \big( (P' \extchoice R') \timeout (Q\intchoice S), (P \timeout Q) \intchoice (R \timeout S) \big),\linebreak
     \big( P \timeout Q , (P \extchoice R) \timeout (Q\intchoice S)\big),
     \big( R \timeout S , (P \extchoice R) \timeout (Q\intchoice S)\big),
     \big( Q\intchoice S, (P \timeout Q) \intchoice (R \timeout S) \big),\linebreak
     \big( S, (P' \extchoice R') \timeout (Q\intchoice S) \big),
     \big( S, R \timeout S \big),
     \big( S, Q \intchoice S \big)
  \mid P' \sqsupseteq_{CS}^\Delta P \wedge R' \sqsupseteq_{CS}^\Delta R\} \cup \Id$ is a divergence-preserving coupled simulation.
\qed
\end{proof}

\begin{proposition}
  Axiom $\ES$ is valid for $\equiv_{CS}^\Delta$.
\end{proposition}
\begin{proof}
  $\{\big( (P \timeout Q) \extchoice (R \timeout S) , (P \extchoice R) \timeout (Q\extchoice S)\big),
     \big( (P \extchoice R) \timeout (Q\extchoice S), (P \timeout Q) \extchoice (R \timeout S) \big),\linebreak
     \big( Q' \extchoice (R \timeout S) , (P \extchoice R) \timeout (Q\extchoice S)\big)~~,~~
     \big( (P \timeout Q) \extchoice S' , (P \extchoice R) \timeout (Q\extchoice S)\big),\linebreak
     \big( Q'\extchoice S', Q' \extchoice (R \timeout S) \big),
     \big( Q'\extchoice S', (P \timeout Q) \extchoice S' \big),
   \mid Q \dto{} Q' \wedge S \dto{} S'\} \cup \Id$ is a divergence-preserving coupled simulation.
\qed
\end{proof}

\begin{proposition}
  Axiom $\EI$ is valid for $\equiv_{CS}^\Delta$.
\end{proposition}
\begin{proof}
  $\{\big( P' \extchoice (Q \intchoice R) , (P \extchoice Q) \intchoice (P \extchoice R) \big),
     \big( (P \extchoice Q) \intchoice (P \extchoice R) , P \extchoice (Q \intchoice R) \big),\linebreak
     \big( P' \extchoice Q , P' \extchoice (Q \intchoice R) \big)
   \mid P \dto{} P'\} \cup \Id$ is a divergence-preserving coupled simulation.
\qed
\end{proof}

\begin{proposition}
  Axiom $\Pru$ is valid for $\equiv_{CS}^\Delta$.
\end{proposition}
\begin{proof}
  $\{\big( (a{\rightarrow} P) \extchoice a{\rightarrow}(P \intchoice Q) , a{\rightarrow}(P \intchoice Q) \big),
     \big( a{\rightarrow}(P \intchoice Q) , (a{\rightarrow} P) \extchoice a{\rightarrow}(P \intchoice Q) \big)\}\linebreak[3]
  \cup \Id$ is a divergence-preserving coupled simulation.
\qed
\pagebreak[3]
\end{proof}

\begin{proposition}
  Axioms {\bf P0--1} and {\bf P4--10} are valid for strong bisimilarity.\\
  Axioms {\bf P11--16} are valid for $\equiv_{CS}^\Delta$.
\end{proposition}
\begin{proof}
   Straightforward.
\qed
\end{proof}

\begin{proposition}
  Axioms $\Us$, $\HI$, {\bf R0--5} and {\bf T0--6}  are valid for strong bisimilarity.
  Axioms {\bf H5--8} are valid for $\equiv_{CS}^\Delta$.
\end{proposition}
\begin{proof}
   Straightforward.
\qed
\end{proof}

\weg{
\begin{proposition}
  Axiom $\US$ is valid for $\equiv_{CS}^\Delta$.
\end{proposition}
\begin{proof}
  $\{\big( (P \timeout Q)\interrupt R' , (P\interrupt R) \timeout (Q\interrupt R) \big),
     \big( (P\interrupt R') \timeout (Q\interrupt R) , (P \timeout Q)\interrupt R \big),\linebreak
     \big( Q\interrupt R' , (P \timeout Q)\interrupt R' \big),
     \big( Q\interrupt R , (P\interrupt R') \timeout (Q\interrupt R)\big)
  \mid R \dto{} R'\} \cup \Id$ is a divergence-preserving coupled simulation.
\qed
\end{proof}
}

\begin{proposition}
  Axiom $\UI$ is valid for $\equiv_{CS}^\Delta$.
\end{proposition}
\begin{proof}
  $\{\big( (P \intchoice Q)\interrupt R' , (P\interrupt R) \intchoice (Q\interrupt R) \big),
     \big( (P\interrupt R) \intchoice (Q\interrupt R) , (P \intchoice Q)\interrupt R \big),\linebreak
     \big( P\interrupt R' , (P \intchoice Q)\interrupt R' \big)
  \mid R \dto{} R'\} \cup \Id$ is a divergence-preserving coupled simulation.
\qed
\end{proof}

\weg{
\begin{proposition}
  Axiom $\UA'$ is valid for strong bisimilarity.
\end{proposition}
\begin{proof}
   $\{\big( P \interrupt Q , \big(\Extchoice_{\!i\in I} (a_i {\rightarrow} (P_i\interrupt Q) )\big) \extchoice Q \big),
      \big( \big(\Extchoice_{\!i\in I} (a_i {\rightarrow} (P_i\interrupt Q) )\big) \extchoice Q , P \interrupt Q \big)
   \mid P = \Extchoice_{\!i \in I}(a_i\rightarrow P_i) \wedge 
        Q = \Extchoice_{j \in J}(b_j\rightarrow Q_j)\} \cup \Id$ is a strong bisimulation.
\qed
\end{proof}
}

\section{Completeness}

Let $\Th$ be the axiomatisation of \tab{axsCS}.
\begin{proposition}\label{pr:nf}
For each recursion-free CSP process $P$ without interrupt operators there is a CSP process $Q$ in
normal form such that $\Th \vdash P=Q$.
\end{proposition}

\begin{proof}
  By structural induction on $P$ it suffices to show that for each $n$-ary CSP operator $Op$, and
  all CSP processes $P_1,...,P_n$ in normal form, also $Op(P_1,...,P_n)$ can be converted to normal form.
  This I do with structural induction on the arguments $P_i$.
\begin{itemize}
\item Let $P=\STOP$ or $\div$. Then $P$ is already in normal form. Take $Q:=P$.
\item Let $P=a \rightarrow P'$. By assumption $P'$ is in normal form; therefore so is $P$.
\item Let $P=P_1 \intchoice P_2$. By assumption $P_1$ and $P_2$ are in normal form.
  So $P\mathbin=\left(\!\big([\div \mathbin{\extchoice]}\Extchoice_{\!i\in I} (a_i\mathbin\rightarrow R_i)\big)\timeout \Intchoice_{\!j\in J}\!R_j\!\right)
  \intchoice
  \left(\!\big([\div\mathbin{\extchoice]}\Extchoice_{\!l\in L} (a_l\mathbin\rightarrow R_l)\big)\timeout \Intchoice_{\!j\in M}\!R_j\!\right)$
  with $R_j=\big([\div \mathbin{\extchoice]}\Extchoice_{k\in K_j}(a_{kj}\rightarrow R_{kj})\big)$ for $j\in J\cup M$.
  With Axiom $\Ss$ I may assume that $J,M\neq\emptyset$.
  Now Axiom $\IS$ converts $P$ to normal form.
\item Let $P=P_1 \extchoice P_2$. By assumption $P_1$ and $P_2$ are in normal form.
  So $P\mathbin=\left(\!\big([\div \mathbin{\extchoice]}\Extchoice_{\!i\in I} (a_i\mathbin\rightarrow R_i)\big)\timeout \Intchoice_{\!j\in J}\!R_j\!\right)
  \extchoice
  \left(\!\big([\div \mathbin{\extchoice]}\Extchoice_{\!l\in L} (a_l\mathbin\rightarrow R_l)\big)\timeout \Intchoice_{\!j\in M}\!R_j\!\right)$
  with $R_j=\big([\div \mathbin{\extchoice]}\Extchoice_{k\in K_j}(a_{kj}\rightarrow R_{kj})\big)$ for $j\in J\cup M$.
  With $\Ss$ I may assume that $J,M\neq\emptyset$.
  Now Axioms $\ES$ and $\EI$ convert $P$ to normal form.
\item Let $P=P_1 \timeout P_2$.
  Axioms \textbf{S2--4} and $\EI$ convert $P$ to normal form.
\item Let $P=P_1 \|_A P_2$.
  Axioms $\Pc$ and \textbf{P4--16}, together with the induction hypothesis, convert $P$ to normal form.
\item Let $P=P\conceal A$.
  Axioms $\HI$ and \textbf{H5--8}, together with the induction hypothesis, convert $P$ to normal form.
\item Let $P=f(P)$.
  Axioms \textbf{R0--5}, together with the induction hypothesis, convert $P$ to normal form.
\item Let $P=P_1 \throw P_2$.
  Axioms \textbf{T0--6}, together with the induction hypothesis, convert $P$ to normal form.
\end{itemize}
\end{proof}

\begin{lemma}\label{lem:saturated}
For any CSP expression $P$ in head normal form there exists a saturated CSP expression $Q$ in head normal form.
\end{lemma}
\begin{proof}
Let $P=\big([\div \mathbin{\extchoice]}\Extchoice_{\!i\in I} (a_i\mathbin\rightarrow R_i)\big)\timeout \Intchoice_{\!j\in J}\!R_j$.
Then $P$ has the form $S \timeout R$. By Axioms \textbf{S1--3}
$\Th\vdash P = (S\extchoice R) \timeout R$. By means of Axioms $\EI$ and $\SI$ the
subexpression $S \extchoice R$ can be brought in the form
$[\div \mathbin{\extchoice]}\Extchoice_{\!l\in L} (a_l\mathbin\rightarrow R_l)$.
The resulting term is saturated.
\qed
\end{proof}

\begin{definition}\rm
A CSP expression \plat{$\big(\Extchoice_{i\in I}(b_i\rightarrow P_i)\big)$} is \emph{pruned} if,
for all $i,h\in I$, $b_i\mathbin=b_h \wedge P_i \sqsupseteq_{CS}^\Delta P_h  \Rightarrow i=h$.
\end{definition}

\begin{theorem}
Let $P$ and $Q$ be recursion-free CSP processes without interrupt operators. Then
$P \equiv_{CS}^\Delta Q$ iff $\Th \vdash P = Q$.
\end{theorem}

\begin{proof}
``$\Leftarrow$'' is an immediate consequence of the soundness of the axioms of $\Th$,
and the fact that $\equiv_{CS}^\Delta$ is a congruence for all operators of CSP\@.

``$\Rightarrow$'':
Let $\depth(P)$ be the length of the longest trace of $P$---well-defined for recursion-free
processes $P$. If $P \equiv_{CS}^\Delta Q$ then $\depth(P)=\depth(Q)$. Given $P \equiv_{CS}^\Delta Q$,
I establish $\Th \vdash P = Q$ with induction on $\depth(P)$.

By \pr{nf} I may assume, without loss of generality, that $P$ and $Q$ are in normal form.
By \lem{saturated} I furthermore assume that $P$ and $Q$ are saturated.
Let $P\mathbin=\big([\div \mathbin{\extchoice]}\Extchoice_{\!i\in I} (a_i\mathbin\rightarrow R_i)\big)\timeout \Intchoice_{\!j\in J}\!R_j$
and $Q\mathbin=\big([\div \mathbin{\extchoice]}\Extchoice_{\!l\in L\!} (a_l\mathbin\rightarrow R_l)\big)\timeout \Intchoice_{\!j\in M}\!R_j$\linebreak
with \plat{$R_j\mathbin=\big([\div \mathbin{\extchoice]}\Extchoice_{k\in K_j}(a_{kj}\rightarrow R_{kj})\big)$}
for $j\mathbin\in J\cup M$, where $R_i$, $R_l$ and $R_{kj}$ are again in normal form.

Suppose that there are $i,h\in I$ with $i\neq h$, $a_i=a_h$ and $R_i \sqsupseteq_{CS}^\Delta R_h$.
Then $R_i \intchoice R_h \equiv_{CS}^\Delta R_h$ by \pr{CSpreorder}.
Since $\depth(R_i \intchoice R_h) < \depth(P)$, the induction hypothesis yields
$\Th \vdash R_i \intchoice R_h = R_h$. Hence
Axiom $\Pru$ allows me to prune the summand $a_i \rightarrow R_i$ from \plat{$\Extchoice_{\!i\in I} (a_i\mathbin\rightarrow R_i)$}.
Doing this repeatedly makes $\Extchoice_{\!i\in I} (a_i\mathbin\rightarrow R_i)$ pruned.
By the same reasoning I may assume that $\Extchoice_{\!l\in L} (a_l\mathbin\rightarrow R_l)$
is pruned.

Since $P{\Uparrow} \Leftrightarrow Q{\Uparrow}$ and $P$ and $Q$ are saturated,
$P$ has the $\div$-summand iff $Q$ does.
I now define a function $f:I\rightarrow L$ such that $a_{f(i)}=a_i$ and
$R_i \sqsupseteq_{CS}^\Delta R_{f(i)}$ for all $i \in I$.

Let $i \in I$. Since $P \goto{a_i} R_i$, by \df{coupled} $Q \dto{a_i} Q'$ for some $Q'$
with \plat{$R_i \sqsupseteq_{CS}^\Delta Q'$}.
Hence either there is an $l\in L$ such that $a_l=a_i$ and $R_l \dto{} Q'$,
or there is a $j\in M$ and $k\in K_j$ such that $a_{kj}=a_i$ and $R_{kj} \dto{} Q'$.
Since $P$ is saturated, the first of these alternatives must apply.
By \pr{tau} \plat{$Q' \sqsupseteq_{CS}^\Delta R_l$} and by \pr{preorder} $R_i \sqsupseteq_{CS}^\Delta R_l$.
Take $f(i):=l$.

By the same reasoning there is a function $g:L\rightarrow I$ such that $a_{g(l)}=a_l$ and
\plat{$R_l \sqsupseteq_{CS}^\Delta R_{g(l)}$} for all $l \in L$.
Since \plat{$\Extchoice_{\!i\in I} (a_i\mathbin\rightarrow R_i)$}
and \plat{$\Extchoice_{\!l\in L} (a_l\mathbin\rightarrow R_l)$}\vspace{1pt}
are pruned, there are no different $i,h \in I$ (or in $L$) with $a_i=a_h$ and $R_i \sqsupseteq_{CS}^\Delta R_h$.
Hence the functions $f$ and $g$ must be inverses of each other.
It follows that $Q=\big([\div \mathbin{\extchoice]}\Extchoice_{i\in I} (a_i\rightarrow R_{f(i)})\big)\timeout \Intchoice_{j\in M}R_j$
with $R_i \equiv_{CS}^\Delta R_{f(i)}$ for all $i \in I$.
By induction $\Th \vdash R_i = R_{f(i)}$ for all $i\in I$.

So in the special case that $I=M=\emptyset$ I obtain $\Th \vdash P = Q$. \hfill (*)

Next consider the case $J=\emptyset$ but $M \neq \emptyset$.
Let $j \in M$. Since \plat{$Q \dto{} R_j$}, there is a $P'$ with \plat{$P \dto{} P'$} and \plat{$R_j \sqsupseteq_{CS}^\Delta P'$}.
Moreover, there is a $P''$ with $P' \dto{} P''$ and \plat{$P'' \sqsupseteq_{CS}^\Delta R_j$}.
Since $J=\emptyset$, $P''=P'=P$, so $P \equiv_{CS}^\Delta R_j$.
By (*) above $\Th \vdash P \mathbin= R_j$. This holds for all $j \in J$, so by Axiom $\Ii$
\plat{$\Th\vdash Q=\big([\div \mathbin{\extchoice]}\Extchoice_{i\in I} (a_i\rightarrow R_i)\big)\timeout P$}.
By Axiom $\Si$ one obtains $\Th\vdash P=Q$.

The same reasoning applies when $M\mathbin=\emptyset$ but $J\mathbin{\neq}\emptyset$.
So henceforth I assume $J,M \mathbin{\neq}\emptyset$.
I now define a function $h\mathop{:}J\mathbin\rightarrow M$ with
$\Th\mathbin\vdash R_j \mathbin=R_{h(j)}$ for all $j \mathbin\in J$.

Let $j \in J$. Since $P \dto{\tau} R_j$, by \df{coupled} $Q \dto{} Q'$ for some $Q'$
with \plat{$R_j \sqsupseteq_{CS}^\Delta Q'$}, and \plat{$Q' \dto{} Q''$} for some $Q''$
with \plat{$Q'' \sqsupseteq_{CS}^\Delta R_j$}.
There must be an $m\in M$ with $Q'' \dto{} R_m$.
By \df{coupled} $R_j\dto{}R'$ for some $R'$ with \plat{$R_m \sqsupseteq_{CS}^\Delta R'$}, and
$R' \dto{} R''$ for some $R''$ with \plat{$R'' \sqsupseteq_{CS}^\Delta R_m$}.
By the shape of $R_j$ one has $R''=R'=R_j$, so $R_j \equiv_{CS}^\Delta R_m$.
By (*) above $\Th \vdash R_j = R_m$. Take $h(j):=m$.

By the same reasoning there is a function $e\mathop:M\rightarrow J$ with
$\Th\mathbin\vdash R_m \mathbin=R_{e(m)}$ for all $m \mathbin\in M$.
Using Axioms \textbf{I1--3} one obtains $\Th\vdash P=Q$.
\qed
\end{proof}

\section{Conclusion}

This paper contributed a new model of CSP, presented as a semantic equivalence on labelled transition
systems that is a congruence for the operators of CSP. It is the finest I could find that allows a
complete equational axiomatisation for closed recursion-free CSP processes that fits within the
existing syntax of the language. For $\tau$-free system, my model coincides with strong
bisimilarity, but in matching internal transitions it is less pedantic than weak bisimilarity.

It is left for future work to show that recursion is treated well in this model,
and also to extend my complete axiomatisation with the interrupt operator of {\sc Roscoe} \cite{Ro97,Ro10}.

An annoying feature of my complete axiomatisation is the enormous collections of heavy-duty axioms
needed to bring parallel compositions of CSP processes in head normal form. These are based on the
expansion law of {\sc Milner} \cite{Mi90ccs}, but a multitude of them is needed due to the optional
presence of divergence-summands and sliding choices in head normal forms.
In the process algebra ACP the expansion law could be avoided through the addition of two auxiliary
operators: the left merge and the communication merge \cite{BK84}. Unfortunately, failures-divergences
equivalence fails to be a congruence for the left-merge, and the same problems exists for any other
models of CSP \cite[Section~3.2.1]{GV93}. In \cite{AI91} an alternative left-merge is proposed, for which
failures-divergences equivalence, and also $\equiv_{CS}^\Delta$, is a congruence. It might be used to
eliminate the expansion law $\PE$ from the axiomatisation of \tab{axsFD}. 
Unfortunately, the axiom that splits a parallel composition between a left-, right- and communication
merge (Axiom CM1 in \cite{BK84}), although valid in the failures-divergences model, is not valid for 
$\equiv_{CS}^\Delta$. This leaves the question of how to better manage the axiomatisation of parallel
composition entirely open.

\newpage

\newcommand{\urlprefix}{}
\bibliographystyle{eptcsini}
\bibliography{roscoe}
\end{document}